\documentclass[11pt]{article}
\pdfoutput=1 
\usepackage[sc]{mathpazo}

\usepackage[margin=1in]{geometry}
\usepackage[english]{babel}
\usepackage[utf8]{inputenc}
\usepackage[compact]{titlesec}
\usepackage{float}

\usepackage{cmap}
\usepackage{bm}
\usepackage{amsfonts}
\usepackage{amsmath}
\usepackage{amssymb}
\usepackage{amsthm}
\usepackage{float}
\usepackage{graphics}

\usepackage{hyperref}
\usepackage[svgnames]{xcolor}
\hypersetup{colorlinks={true},urlcolor={blue},linkcolor={DarkBlue},citecolor=[named]{DarkGreen},linktoc=all}
\usepackage{natbib}

\usepackage{microtype}
\usepackage[capitalise,nameinlink,noabbrev]{cleveref}

\usepackage{doi}

\usepackage{tikz}

\newcommand*\circled[1]{\tikz[baseline=(char.base)]{
            \node[shape=circle,draw,inner sep=2pt] (char) {#1};}}
\usepackage{booktabs} 
\usepackage{amsfonts,amsbsy,amsthm}
\usepackage{thmtools,thm-restate}
\usepackage{enumerate}
\usepackage{bbm}
\usepackage{nicefrac}
\usepackage{wrapfig}
\usepackage[justification = centering]{subcaption}
\usepackage{mathtools}
\usepackage{array,multirow}

  {\list{}{\leftmargin=#1\rightmargin=#1}\item[]}%
  {\endlist}

\newtheorem{theorem}{Theorem}[section]
\newtheorem{lemma}[theorem]{Lemma}

\newtheorem{proposition}[theorem]{Proposition}
\newtheorem{corollary}[theorem]{Corollary}

\newtheorem{example}[theorem]{Example}

\usepackage{tikz}
\usetikzlibrary{calc, arrows,matrix,decorations.pathreplacing}
\usepackage{paralist}
\usepackage{algorithm}
\usepackage{algorithmic}
\usepackage{bbding}
\usepackage{pifont}
\usepackage{fontawesome}

\usepackage[margin=1in]{geometry}
\usepackage{tabularx}

\newcommand\tikzmark[2]{%
\tikz[remember picture,baseline] \node[inner sep=2pt,outer sep=0] (#1){#2};%
}

\newcommand\link[2]{%
\begin{tikzpicture}[remember picture, overlay, >=stealth, shift={(0,0)}]
  \draw[red, ->, bend right=80] (#1) to (#2);
\end{tikzpicture}%
}

\usepackage[capitalise,noabbrev]{cleveref}
\Crefname{claim}{Claim}{Claims}
\Crefname{corollary}{Corollary}{Corollaries}
\Crefname{definition}{Definition}{Definitions}
\Crefname{example}{Example}{Examples}
\Crefname{lemma}{Lemma}{Lemmas}
\Crefname{property}{Property}{Properties}
\Crefname{proposition}{Proposition}{Propositions}
\Crefname{remark}{Remark}{Remarks}
\Crefname{theorem}{Theorem}{Theorems}
\title{Equitable Allocation for Mixtures of Goods and Chores}
\author{
	\begin{tabular}{m{0.12\textwidth}m{0.12\textwidth}m{0.12\textwidth}m{0.12\textwidth}m{0.12\textwidth}m{0.12\textwidth}}
		\multicolumn{3}{c}{\textbf{Hadi Hosseini}} & \multicolumn{3}{c}{\textbf{Aditi Sethia}}\\
		\multicolumn{3}{c}{\small{Pennsylvania State University}} & \multicolumn{3}{c}{\small{Indian Institute of Technology, Gandhinagar}}\\
		\multicolumn{3}{c}{\href{mailto:test@test.com}{\small{\texttt{hadi@psu.edu}}}} & \multicolumn{3}{c}{\href{mailto:test@test.com}{\small{\texttt{aditi.sethia@iitgn.ac.in}}}}\\
	\end{tabular}
}

\date{}

\begin{document}

\maketitle

\begin{abstract}
Equitable allocation of indivisible items involves partitioning the items among agents such that everyone derives (almost) equal utility.
We consider the approximate notion of \textit{equitability up to one item} (EQ1) and focus on the settings containing mixtures of items (goods and chores), where an agent may derive positive, negative, or zero utility from an item.
We first show that---in stark contrast to the goods-only and chores-only settings---an EQ1 allocation may not exist even for additive $\{-1,1\}$ bivalued instances, and its corresponding decision problem is computationally intractable. 
We focus on a natural domain of normalized valuations where the value of the entire set of items is constant for all agents.
On the algorithmic side, we show that an EQ1 allocation can be computed efficiently for 
(i) symmetric tri-valued ($\{-1, 0, 1\}$) normalized valuations,
(ii) objective but non-normalized valuations,
(iii) two agents with type-normalized valuations.
Previously, EQX allocations were known to exist only for 2 agents and objective valuations, while the case of subjective valuations remained computationally intractable even with two agents. We make progress by presenting an efficient algorithm that outputs an EQX allocation for symmetric bi-valued ($\{-1,1\}$) normalized subjective valuations for any number of agents. We complement our study by providing a comprehensive picture of achieving EQ1 allocations in conjunction with economic efficiency notions such as Pareto optimality and social welfare.

\end{abstract}



\section{Introduction}
The distribution of indivisible items fairly among a set of agents is a prominent problem within the field of economics and computation \citep{M04fair,10.5555/3033138,brams1996fair,VARIAN197463}.
A host of real-world problems deal with the distribution of items that are considered `goods' (positively valued) by some agents while `chores' (negatively valued) by others: some students may enjoy taking a rigorous course on Artificial Intelligence (AI), while others may find it daunting. 
Paper reviewing is often rewarding when it aligns with the reviewer's expertise, but can become a tedious task (a chore) otherwise.
Given that agents may have different subjective valuations over the items, the goal is to find an allocation of the items to the agents that is fair while satisfying some notion of economic efficiency.

The integration of modern AI models (e.g. pre-trained Large Language Models) into large-scale decision-making pipelines has become increasingly prevalent, with applications spanning both socio-economic domains and logistical optimization problems.
Examples range from recommender systems in streaming platforms to matching mechanisms for course allocation, such as ScheduleScout,\footnote{ScheduleScout: \url{https://www.getschedulescout.com/}} a platform that has facilitated the allocation of over 400,000 course seats while incorporating satisfaction and fairness criteria.
As such systems increasingly mediate the allocation of scarce resources, ensuring equitable outcomes and embedding fairness considerations into their design becomes a critical objective.

The underlying preferences in the above examples can be either positive or negative. The recommender systems, for revenue optimization, may sometimes suggest options that do not align with a user's preferences. A student may have to be allocated a course that he considers a challenge in order to satisfy the capacity constraints.
Therefore, to ensure that no individual/group is systematically disadvantaged in opportunities and outcomes, training equitable and ethical models becomes imperative.


A desirable fairness notion is \textit{equitability}, which requires that the \textit{subjective} value of every agent over their bundle to be the same \citep{Dubins1961}. In other words, each agent is equally `happy' (or unhappy) with its own bundle. There have been other well-studied fairness notions like \emph{envy-freeness}, where every agent values its own bundle better than anyone else's \citep{GS58puzzle,F67resource}.
In a series of studies with human subjects, equitability has been shown to have a major impact on perceived fairness of allocations \citep{HERREINER2010238,Herreiner2009}, and plays a crucial role in practical applications such as divorce settlements \citep{brams1996fair} and rental harmony \citep{10.1145/3131361}. 


While an equitable allocation always exists for divisible items (aka \textit{cake-cutting}) \citep{alon1987splitting,brams2006better,CECHLAROVA2013239}, when items are indivisible (can not be fractionally assigned), equitable solutions may not exist: consider one good (or chore) and two agents who like (or dislike) it.
This impossibility in the presence of indivisible items has motivated `up to one' style approximations. 
In particular, equitability up to any item (EQX) and equitability up to one item (EQ1) are relaxations that allow for inequality between two agents to be resolved by removing any item (in the case of EQX) or some item (in the case of EQ1) from one agent's bundle (see \Cref{sec:prelims} for formal definitions).



In the goods/chores-only settings, existence and computation of EQX/EQ1 allocations have been well-explored; such an allocation always exists and can be computed efficiently \citep{10.5555/3006652.3006719,10.5555/3398761.3398810}. The setting of mixtures is relatively underexplored. The aim of this exposition is to address this gap.
%
%
In contrast to the universal existence of EQX allocation for goods/chores only settings, even the weaker concept of EQ1 allocation may not exist for mixtures of goods and chores, even with two agents and two items, as illustrated in the following example.

\begin{example}
\label{noEQ1}
Consider an instance with two agents (Alice and Bob) and two items such that Alice values the two items at $-1$ each while Bob values them at $1$ each. In any allocation, Alice's valuation is strictly less than that of Bob even after the hypothetical removal of some item from any of the bundles. Therefore, no allocation in this instance is EQ1.
\end{example}

The above existential issue can be primarily attributed to the pronounced `disparity' in how agents subjectively evaluate the entire set of items. In \Cref{noEQ1}, Alice's value for the entire set is $-2$, while Bob's value is $2$.
This observation motivates the study of \textit{normalized valuations} where the value of the entire set of items is a constant (positive, negative, or zero) for all agents.
Normalization is a standard and reasonable assumption in fair division literature, and has been studied in the context of cake-cutting\footnote{In fair cake-cutting, the common assumption is that each agent values the entire cake as $1$ (or $-1$ for burnt cakes)}\citep{10.5555/3033138}, price of fairness \citep{10.1007/s00224-011-9359-y,10.1007/978-3-031-43254-5_16} and online fair division as well \citep{article}.


\subsection{Our Contributions}
Our work provides a deep dive into the existence and computational boundaries of EQ1 allocations for mixtures of goods and chores with and without economic efficiency notions (e.g. Pareto optimality (PO) and utilitarian/egalitarian social welfare; see \Cref{sec:prelims} for formal definitions). We not only show non-existence under non-normalized valuations, but also establish the hardness of deciding the existence of an EQ1 allocation (\Cref{thm:EQ1hard}).
This result complements a recent study that shows computing an EQX allocation---a strengthening of EQ1---is strongly NP-hard \citep{barman2024nearly}.



\Cref{tab:result_stable} provides a summary of our algorithmic results. 

\paragraph{Objective valuations.} In Section \ref{sec:EQ1},
we show that an EQ1 allocation always exists and can be computed efficiently for non-normalized objective valuations, where each item is either a good or a chore for all the agents (\Cref{thm:obj}). Our analysis gives rise to \cref{lem:completion}, which enables the design of new algorithms along with efficiency. It states that if there exists a partial EQ1 allocation of subjective items, it can be completed by allocating the remaining objective items in an EQ1 manner. We note here that objective valuations have also been studied under the name of doubly monotone valuations \citep{bhaskar_et_al:LIPIcs.APPROX/RANDOM.2021.1} in the context of envy-freeness.

\paragraph{Normalized valuations.}
In \Cref{sec:EQ1norm}, we show that under symmetric bi-valued normalized valuations, where an item takes values from $\{-1,1\}$, an EQX allocation always exists and can be computed in polynomial time for any number of agents (\cref{thm:norm_EQ1}). Note that EQX is known to exist only for $2$ agents with objective valuations and is hard to compute for non-objective normalized valuations, even for two agents \citep{barman2024nearly}. Towards tractability, our algorithm gives an EQX allocation for any number of agents with symmetric and normalized valuations, not necessarily objective. Our techniques involve multiple careful transfers of items among the agents and maintain a partial EQ1/EQX allocation at every step.
We additionally show the existence and efficient computation of EQ1 allocations for symmetric tri-valued valuations ($\{-1, 0, 1\}$).
Such valuations generalize binary preferences, capture realistic scenarios involving approvals/dis-approvals/neutrality, and have been studied in the context of  EFX$+$PO \citep{10.1007/978-3-030-58285-2_1}. It is also relevant to note that binary preferences capture many real-world domains, and many significant contributions in allocation and voting problems have solely focused on such preferences
\citep{BOGOMOLNAIA2005165,articlebabaioff,HPP+20fair,10.1007/978-3-030-94676-0_21,DARMANN2015548}. We also note that all our results for $\{-1,0,1\}$ valuations are scale-free and hence, hold for $\{-w,0, w\}$ valuations as well, where $w \in \mathbb{R}$.


Furthermore, we show that when valuations are in addition \textit{type-normalized} (i.e. for all the agents, the sum of the value of all the goods (or chores) is a constant $g$ (or $c$)), an EQ1 allocation can be computed efficiently for two agents (\Cref{thm:type_n=2}).


\paragraph{Efficiency.}
We also consider EQ1 in conjunction with the efficiency notion of Pareto optimality (PO). An allocation is said to be Pareto optimal if no reallocation of items can make any agent better off without making some other agent worse off. In \Cref{sec:EQ1$+$PO}, we show that an EQ1 and PO allocation may not exist even with $\{-1, 1\}$ normalized valuations.
Additionally, we show that deciding whether an EQ1 and PO allocation exists is intractable even for the special class of type-normalized valuations (\Cref{thm:eq1pohard}). 
Nonetheless, we develop a polynomial-time algorithm for $\{-1,0,1\}$ normalized valuations that computes an EQ1$+$PO allocation, when one exists (\Cref{lem:typen101}). We note here that the presence of zeros has been a source of computational challenge for EQ1$+$PO allocations even in the only goods setting \citep{10.5555/3367032.3367073,DBLP:journals/jair/GargM24}. Our algorithmic technique involves reducing the instance to a partial instance such that every item is non-negatively valued by at least one agent, computing a Nash optimal solution, and invoking \Cref{lem:completion} to assign any remaining chore.

\paragraph{Welfare.}
In \Cref{sec:EQ1+welfare}, we develop a pseudo-polynomial time algorithm for finding EQX allocations that maximize utilitarian or egalitarian welfare (see \Cref{sec:prelims} for formal definitions.) To that end, we use a dynamic programming approach that keeps track of a set of states representing a set of all possible allocations (\Cref{thm:psuedo}).
Note that maximizing welfare along with approximate equitability is already known to be weakly NP-Hard, even when all the items are goods \citep{10.5555/3545946.3598685}.
In Appendix D, we provide additional results on the compatibility of EQ, EQ1, and PO, with envy-freeness and its relaxations.

\begin{table*}[t]
\small
\centering
\begin{tabular}{@{}llllll@{}}
\toprule
     &  & Arbitrary                                                                                            & Objective                                                                          & Normalized                                                                              & Type-normalized                                                                                \\ \midrule
EQ1 & \begin{tabular}[c]{@{}l@{}} Exist. \\ Comp. \end{tabular}  & \begin{tabular}[c]{@{}l@{}} \ding{55}, $\{-1,1\}$ (Ex \ref{noEQ1})\\ (w) NP-c (Thm \ref{thm:EQ1hard})\end{tabular} & \begin{tabular}[c]{@{}l@{}}\ding{51}\\ P (Prop \ref{thm:obj}) \end{tabular}             & \begin{tabular}[c]{@{}l@{}}\ding{51}, $\{-1, 0, 1\}$\\ P $\{-1, 0, 1\}$ (Thm \ref{thm:norm_EQ1})\end{tabular} & \begin{tabular}[c]{@{}l@{}}\ding{51}, n = 2\\ P (Thm \ref{thm:type_n=2})\end{tabular}               \\
\midrule
EQ1$+$PO & \begin{tabular}[c]{@{}l@{}} Exist. \\ Comp. \end{tabular} & \begin{tabular}[c]{@{}l@{}}\ding{55} $\{-1,1\}$ (Ex \ref{ex:pos})\\ (s) NP-h (Thm \ref{thm:eq1pohard})\end{tabular}              & \begin{tabular}[c]{@{}l@{}}\ding{55}\\ (s) NP-h$^{\star}$\end{tabular} & \begin{tabular}[c]{@{}l@{}}\ding{55}, $\{-1, 1\}$\\ P, $\{-1, 0, 1\}$ (Thm \ref{lem:typen101})\end{tabular} & \begin{tabular}[c]{@{}l@{}}\ding{51}, n = 2, $\{-1, 0, 1\}$ (Thm \ref{lem:type101})\\ P, $\{-1, 0, 1\}$ (Thm \ref{lem:typen101}) \end{tabular} \\ \bottomrule
\end{tabular}
    \vspace{0.3cm}

\caption{The summary of our results for arbitrary (non-normalized), objective (non-normalized), normalized, and type-normalized valuations. The result marked by $\star$ follow from \citep{10.5555/3367032.3367073} respectively. P, (w) NP-c, and (s) NP-h denote polynomial algorithm, weak NP-completeness, and strong NP-hardness. 
All EQ1 allocations (including welfare-maximizing) can be computed in pseudo-polynomial time (\cref{thm:psuedo}). Deciding the existence of EQ1$+$PO allocations is strongly NP-hard even for type-normalized valuations (\Cref{thm:eq1pohard}).
}
  \vspace{0.5cm}
\label{tab:result_stable}
\end{table*}

\subsection{Related Work}



\paragraph{Equitability and its Relaxations.}
Equitability as a fairness notion was first studied in the context of divisible items, where it is possible to assign a fraction of an item to any agent \citep{Dubins1961}. It is known that equitable allocations always exist for divisible items \citep{CECHLAROVA2013239, CHEZE201792}. There is a computational limitation to finding such an allocation \citep{10.1145/3033274.3085107}, but approximate equitable allocations admit efficient algorithms \citep{CECHLAROVA2012249, doi:10.1080/02331934.2011.563306}. In the setting of indivisible items, \cite{10.5555/3367032.3367073, 10.5555/3398761.3398810} studied equitability along with efficiency guarantees for goods/chores only settings. Follow-up work has studied equitability in various contexts: when items lie on a path, Follow-up work has studied equitability in various contexts: when items lie on a path, the trade-off between equitability and welfare (the price of equitability), and its interaction with welfare guarantees  \citep{DBLP:journals/jair/GargM24,SUKSOMPONG2019227,DBLP:journals/corr/abs-2101-09794,10.1007/s00224-011-9359-y, 10.1145/2781776, 10.5555/3545946.3598685, 10.1007/978-3-031-43254-5_16}. 

\paragraph{Equitability under non-monotone valuations.} Perhaps closest to our work is that of \cite{barman2024nearly}, who study EQX allocations for mixtures and exhibit an algorithm for two agents and objective additive valuations and hardness for non-objective valuations even with two agents.

Subsequent to our work, \cite{bilò2025} also noted that EQ1 allocations exist for objective valuations. However, that is the only point of overlap with our work. They studied weaker variants of EQX* and EQ1* in the context of non-monotone valuations. While EQX
requires the fairness property to hold regardless of which item is removed, EQX* allows the property to hold for the removal of either goods or chores. EQ1* entails removal of one good from the richer bundle \emph{and} one chore from the poor bundle. They showed that both EQX* and EQ1* allocations always exists and can be computed in polynomial time.

A similar notion to equitability, named jealousy-freeness (JF), and its up-to-one-style-relaxations JFX and JF1 have also been considered, wherein, in addition to requiring that the hypothetical removals are restricted to non-zero valued items, these definitions also allow hypothetical addition of an item to a bundle \citep{Alek2020}.
\cite{Alek2023} segregated the notions where addition/duplication is allowed and renamed these notions to 
DJFX and DJF1 (that is, JFX and JF1 up to removal or duplication). Their JFX definition is analogous to our EQX definition, albeit they restrict the removals to non-zero valued items. However, we note that their Theorem 3 does not hold true (See \Cref{sec:additionalrelatedwork} for details).



 

\paragraph{Other Fairness Notions.}
Fairness notions like envy-freeness (every agent values its own bundle more than it values anyone else's) \citep{berczi2020envy,hosseini2023fairlydividing,hosseinietalFairlyAllocating23} and maximin shares (a threshold based criterion where every agent gets the minimum value it can guarantee itself after dividing the set of items into number of agents many bundles and choosing the worst of them) \citep{10.1145/3465456.3467553,DBLP:conf/aaai/KulkarniMT21} have also been studied in the context of mixtures \citep{10.1613/jair.1.15800}. A brief overview of some of the known results in these settings is as follows.

For goods and even monotone valuations, envy-free upto one item (EF1)\footnote{See \Cref{sec:EF+EQ} for formal definition, which is analogous to EQ1.} allocations are always attainable and can be found efficiently via a graph-theoretic envy-cycle elimination algorithm of~\cite{LMM04}. In the chores setting, the standard envy-cycle elimination fails;~\cite{Bhaskar2020OnAE} overcame this by resolving a top-trading envy cycle, ensuring EF1. For mixed manna,~\cite{azizetalFairallocation22} established EF1 existence and efficient computation using a double round-robin method.

EF1 is also known to be compatible with PO \citep{10.1145/3355902} and also admits a pseudo-polynomial time algorithm \citep{10.1145/3219166.3219176} for goods. For chores, EF1+PO was earlier known to exist only for two agents \citep{azizetalFairallocation22}, three agents \citep{gargetalNewAlgorithms23}, and three agent-types \citep{gargetalWeightedEF124}; and other restricted classes of valuations~\citep{ebadianetalHowfairly22,gargetalFairEfficientAAAI22,wuetalWeightedEF123,GargConstantFactorSTOC25,azizetalFairallocation23}. The case of four or more agents is very recently settled by \cite{mahara2025}. For mixtures, the joint existence of EF1 and PO allocations remains unresolved even for three agents with additive valuations, except some restricted domains like lexicographic preferences \citep{hosseinietalFairlyAllocating23}, identical valuations, and ternary ($\{-\alpha, 0, \beta\}$) valuations \citep{10.1007/978-3-030-58285-2_1}. For two agents, EF1+PO can be achieved via a discrete adjusted winner rule \citep{azizetalFairallocation22}. Recently, \cite{barman2025} showed that envy-freeness up to $k$ reallocations (EFR-$k$) is compatible with PO for any number of agents.

The Maximin Share (MMS) value of an agent is the value that it can guarantee itself
if it partitions the set of items into bundles equal to the number of agents, given that it is the last agent to choose its bundle. An allocation that guarantees every agent its MMS value is said to be an MMS allocation. The problem is well-studied when there are only
goods \citep{10.1145/3391403.3399526,10.1145/3140756} or chores \citep{barman2020approximation,10.1145/3465456.3467555}. Notably, an MMS allocation may not even exist for goods with additive valuations \citep{10.1145/3140756,feige2021tight}, but a Polynomial Time Approximation Scheme (PTAS) to compute the MMS values
of agents is well-known \citep{WOEGINGER1997149}. In contrast, for the mixed manna, \cite{10.1145/3465456.3467553} showed that
finding even an approximate MMS value of an agent up to
any approximation factor in (0,1] is NP-hard. Complementing the hardness, \cite{DBLP:conf/aaai/KulkarniMT21} gave a 
PTAS to compute the MMS value, when its absolute value is at least $\frac{1}{p}$ times either the total value of all
the goods or total cost of all the chores, for some constant $p \geq 1$. Approximate MMS with PO is also studied, and the problem of computing an $\alpha$-MMS+PO allocation is known to be tractable under the following two conditions: (1) the number of agents is a constant, and (2)
for every agent, its absolute value for all the items is at least a constant factor of its total (absolute) value
for all the goods or all the chores \citep{10.1145/3465456.3467553}.



   
\section{Preliminaries}
\label{sec:prelims}

\paragraph{Setting.} A fair division instance ($[n], [m], \mathcal{V}$) consists of $n \in \mathbb{N}$ agents, $m \in \mathbb{N}$ items and valuations $V = \{v_1, v_2, \ldots v_n\}$ where $v_i : 2^{[m]} \rightarrow \mathbb{Z}$ captures the value that an agent $i$ derives from a subset of items. We restrict our attention to \textit{additive valuations} where for any subset $S \subseteq [m]$, we have $v_i(S) = \sum_{o \in S} v_i\{o\}$. 
Let $O$ be the set of all $m$ items. We say that an item is an \textit{objective good} if it is valued non-negatively by all the agents. That is, $o \in O$, such that $v_i(o) \geq 0$ for all $i \in [n]$, and denote the set of all objective goods by $O^+$. Similarly, all items $o \in O$ such that $v_i(o) \leq 0$ are called \textit{objective chores}, denoted by $O^-$. The remaining objects are called \textit{subjective items}, denoted by $O^\pm$. An instance where every item is either an objective good or an objective chore is said to be an instance with objective valuations, otherwise, it is said to be an instance with subjective valuations.

The valuations $\{-1,0, 1\}$ denote that each item is either valued at a constant $1$, $0$, or $-1$ by every agent and are said to be \emph{symmetric tri-valued valuations}. Likewise, $\{-1,1\}$ are called \emph{symmetric bi-valued valuations}. It is relevant to note that all our results for symmetric bi-valued/tri-valued valuations are scale-free and hence, hold for $\{-w, w\}$ and $\{-w, 0, w\}$ valuations as well, where $w \in \mathbb{R}$.

An allocation $A = (A_1, A_2, \ldots A_n)$ is a partition of $m$ items into $n$ bundles, one for each agent. An allocation is said to be complete if no item remains unallocated, otherwise it is a partial allocation. The focus of this work is to compute complete allocations.

\paragraph{Equitability.}
An allocation $A$ is said to be \textit{equitable} (EQ) if all the agents derive equal value from their respective bundles, that is, for every pair of agents $i$ and $j$, we have $v_i(A_i) = v_j(A_j)$. An allocation $A$ is said to be \textit{equitable up to one item} (EQ1) if for any pair of agents $i$ and $j$ such that $v_i(A_i) < v_j(A_j)$, either there exists some good $g$ in $A_j$ (a non-negatively valued item by $j$) such that $v_i(A_i) \geq v_j(A_j \setminus \{g\})$ or there exists some chore $c$ in $A_i$ (a non-positively valued item by $i$) such that $v_i(A_i \setminus \{c\}) \geq v_j(A_j)$. Further, an allocation $A$ is said to be \textit{equitable up to any item} (EQX) if for any pair of agents $i$ and $j$ such that $v_i(A_i) < v_j(A_j)$, we have $v_i(A_i) \geq v_j(A_j \setminus \{g\})$ for all goods $g$ in $A_j$ (non-negatively valued items by $j$), and $v_i(A_i \setminus \{c\}) \geq v_j(A_j)$ for all chores $c$ in $A_i$ (non-positively valued items by $i$).

For ease of exposition, we sometimes use `poor' and `rich' to identify an arbitrary agent. Formally, an agent $i$ is \textit{poor} under an allocation $A$ if $v_i(A_i) \leq v_j(A_j), \forall~j (\neq i) \in [n]$ and it is \textit{rich} if $v_i(A_i) > v_j(A_j), \forall~j (\neq i) \in [n]$.

\paragraph{Efficiency.} An allocation $A$ is said to be Pareto dominated by another allocation $A^\prime$ if $v_i(A^\prime_i) \geq v_i(A_i)$ for all agents $i$ and there exists at least one agent $j$ such that $v_j(A^\prime_j) > v_j(A_j)$. If an allocation is not Pareto dominated by any other allocation, then it is said to be Pareto optimal (PO). 
The  \textit{utilitarian welfare} of an allocation $A$ is the arithmetic mean of individual values under $A$, denoted by $UW(A) := \frac{1}{n} \sum_{i \in [n]} v_i(A_i)$. The \textit{egalitarian welfare} is the minimum value among all the individual values, denoted by $EW(A) := \min_{i \in [n]} v_i(A_i)$. The \textit{Nash welfare} is the geometric mean of the individual values, denoted by $NW(A) := (\prod_{i \in [n]} v_i(A_i))^\frac{1}{n}$
We denote by UW, EW and NW those allocations that maximize the utilitarian, egalitarian, and Nash welfare, respectively.

\paragraph{Normalization(s).} 
We consider two types of normalizations. 
An instance is \textit{normalized} if all agents value the entire bundle at a constant, that is, $v_i(O)$ is a constant for every agent $i \in \mathbb{N}$.
An instance is \textit{type-normalized} if for every agent $i$, all the goods sum up to a constant, say $g$, and all the chores sum up to a constant, say $c$, that is, $\sum_{o:v_i(o)\geq0} v_i(o) = g$ and $\sum_{o:v_i(o)\leq0} v_i(o) = c$. 
Note that a type-normalized instance is necessarily normalized (but the converse may not be true).

\section{EQ1 Allocations}
\label{sec:EQ1}

When the input instance has subjective items, an EQ1 allocation may not exist, even for two agents and two items, as illustrated by \Cref{noEQ1}. Moreover, we show in the following result that deciding if there is an EQ1 allocation is weakly NP-Hard. This is in contrast to the only goods or only chores setting where even for non-normalized instances, computing an EQ1 allocation admits simple and efficient algorithms \citep{10.5555/3006652.3006719,10.5555/3398761.3398810}.

\begin{restatable}{theorem}{thmEqOnehard}
\label{thm:EQ1hard}
For any mixed fair division instance, deciding the existence of an EQ1 allocation is (weakly) NP-complete. 
\end{restatable}

\begin{proof}
We exhibit a reduction from $2$-partition, where given a multiset $U = \{b_1, b_2 \ldots b_m\}$ of integers with sum $2T$, the task is to decide if there is a partition into two subsets $S$ and $U \setminus S$ such that sum of the numbers in both the partitions equals $T$. 
We construct the reduced allocation instance as follows. We create $4$ agents, $m$ set-items $\{o_1,\ldots,o_m\}$, and $4$ dummy items $\{d_1, \ldots d_4\}$. The first two agents value the set-items at $\{b_1, \ldots b_m\}$ and all the dummy items at $-3T$. The last two agents value the set items at $0$ and the dummy items at $T$. This completes the construction (see \Cref{tab:EQ1ishard}).

\begin{table}[H]
\centering
\begin{tabular}{c|ccccccccc}
     & $o_1$ & $o_2$ & $\ldots$ & $o_m$  & $d_1$ & $d_2$ & $d_3$ & $d_4$\\ \hline 
   $a_1$  & $b_1$ & $b_2$ & $\ldots$ & $b_m$ & $-3T$ & $-3T$ & $-3T$ & $-3T$ \\
 $a_2$ & $b_1$ & $b_2$ & $\ldots$ & $b_m$ & $-3T$ & $-3T$ & $-3T$ & $-3T$ \\
   $a_3$  & $0$ & $0$  & $\ldots$ & $0$ & $T$ & $T $& $T$& $T$ \\
   $a_4$  & $0$ & $0$  & $\ldots$ & $0$ & $T$ & $T $& $T$& $T$\\
\end{tabular}
\vspace{0.4cm}
\caption{Reduced instance as in the proof of \Cref{thm:EQ1hard}.}
\vspace{0.4cm}
\label{tab:EQ1ishard}
\end{table}

We now argue the equivalence. 

\paragraph{Forward Direction.} Suppose that the $2$-partition instance is a yes-instance. Let $S$ and $U \setminus S$ be the said partitions. Then the allocation where agent $a_1$ gets $S$, $a_2$ gets $U \setminus S$, $a_3$ gets $\{d_1, d_2\}$ and $a_4$ gets $\{d_3, d_4\}$ is clearly an EQ1 allocation.

\paragraph{Reverse Direction.} Suppose there is an EQ1 allocation, say $A$, under the reduced instance. Then, note that $A$ assigns at most one dummy item to the agents $a_1$ and $a_2$, each. If not, say WLOG, $a_1$ receives $\{d_1, d_2\}$ under $A$, then the maximum utility $a_1$ can derive is $-4T$ (where it gets all the set items as well). But, $a_1$ violates EQ1 with respect to $a_3$ and $a_4$, whose minimum utility is $0$ each, and $a_1$ derives negative utility even if it removes the item $d_1$ from its bundle. Therefore, $a_1$ and $a_2$ can get at most one dummy item each. 
We now consider the following cases:
\begin{itemize}
\item $a_1$ and $a_2$ do not receive any dummy item. Note that since $A$ is EQ1, anyone from $a_3$ and $a_4$ cannot get three dummy items, else either one of $a_1$ or $a_2$ violates EQ1. Therefore, $a_3$ and $a_4$ both get two dummy items each. This forces $a_1$ and $a_2$ to receive a utility of at least $T$, thereby corresponding to a partition.

\item $a_1$ receives $d_1$ and $a_2$ receives $d_2$. Note that if the remaining two dummy items are allocated to any one agent, say $a_3$, then EQ1 is violated. Indeed, the maximum utility of $a_1$ and $a_2$  is negative, and even if they choose to ignore any item, they fall short of the utility derived by $a_3$ (which is $2T$). Therefore, $a_3$ receives $d_3$ and $a_4$ receives $d_4$ under $A$. This forces $a_1$ and $a_2$ to receive a utility of $T$ each from the set items, thereby forcing a partition.

\item $a_1$ receives $d_1$ and $a_2$ does not receive any dummy item. Then, WLOG, $a_3$ gets one dummy item and $a_4$ gets two dummy items. To be consistent with EQ1 against $a_4$, agent $a_1$ must get all the set items, thereby deriving a utility of $-3T+2T$. This leaves $a_2$ empty-handed and hence, it violates EQ1 with respect to $a_4$.
Therefore, since $A$ is EQ1, this case does not arise.
\end{itemize}
This concludes the argument.
\end{proof}

\subsection{Objective Valuations}

Given the negative existence and computational results, we identify tractable instances for which an EQ1 allocation always exist and can be computed efficiently. In particular, we develop an algorithm that computes an EQ1 allocation when all items are either objective goods or objective chores. 
This finding gives rise to an intriguing relationship between partial and complete EQ1 allocations, as we state in \Cref{lem:completion}.


\begin{restatable}{proposition}{propobj}
\label{thm:obj}
    For any mixed fair division instance with objective valuations, an EQ1 allocation always exists and can be computed in polynomial time. 
\end{restatable}

\begin{proof} The algorithm iteratively picks a poor agent $p$---breaking ties arbitrarily---and allocates the most valuable good among the remaining items in $O^+$ according to $p$'s valuation. Once $O^+$ is exhausted, it picks a rich agent $r$, who then receives its most disliked chore from $O^-$. The correctness argument is deferred to the appendix.
\end{proof}

We note here that the existence of EQX allocation remains open even for objective valuations for more than two agents. \Cref{thm:obj} fails to guarantee the existence of EQX allocations even for identical objective valuations, as illustrated in the following example. 

\begin{example} Consider the following instance. The output of the algorithm in \Cref{thm:obj} is highlighted, which fails EQX. Indeed, if Alice chooses to ignore the item $o_2$ from Bob's bundle and hypothetically reduces its utility from $1$ to $-1$, it still remains a poor agent. But this allocation is EQ1, since Alice can ignore the chore $o_7$ from its bundle, thereby deriving equal utility as Bob, that is, $1$.

\begin{table}[H]
\footnotesize
    \centering
    \begin{tabular}{c|ccccccc}
         & $o_1$ & $o_2$ & $o_3$ & $o_4$ & $o_5$ &$o_6$ & $o_7$ \\ \hline 
        Alice & \circled{$2$}  & $2$  &  \circled{$2$}  & $2$ &  \circled{$-3$} & $-3$ &  \circled{$-3$} \\
        Bob & $2$  &  \circled{$2$} & $2$ &  \circled{$2$} & $-3$ &  \circled{$-3$} & $-3$ \\
    \end{tabular}
\end{table}
\label{ex:notEQX}
\end{example}


We now show that if there exists a partial EQ1 allocation that allocates $O^\pm$ entirely, then it can always be extended to a complete EQ1 allocation. This lemma will be useful in developing new algorithms in Section \ref{sec:EQ1norm}.

\begin{lemma}[\textbf{Completion Lemma}] 
\label{lem:completion}
Consider a partial EQ1 allocation $A$ that allocates a subset of items, say $S \subset [m]$. Then, $A$ can be completed while preserving EQ1 if $O^\pm \subseteq S$.

\end{lemma}

\begin{proof} Let $A$ be the partial EQ1 allocation that allocates $S$ such that $O^\pm \subseteq S \subset [m]$. 
Since the remaining items are all objective items, we iteratively allocate a poor agent under $A$ its most liked item from $O^+$. Once $O^+$ is exhausted, we iteratively allocate a rich agent its most disliked item from $O^-$. The correctness follows by a similar argument as in \Cref{thm:obj}.
\end{proof}

\section{Normalized Valuations}
\label{sec:EQ1norm}

In the further quest for tractability, we now turn our attention to subjective but normalized and type-normalized valuations.

\subsection{Symmetric Valuations}
\label{sec:symmetricvaluations}

\begin{restatable}{theorem}{thmEQoneww}
\label{thm:EQ1n}
    For any mixed fair division instance with symmetric bi-valued normalized valuations ($\{-1, 1\}$), an EQX (hence, EQ1) allocation always exists and can be computed in polynomial time.
\end{restatable}

It is easy to see that for $\{-1,1\}$ valuations, EQ1 and EQX allocations coincide.
Since all the goods in a rich agent's bundle are valued at $1$ and all the chores in a poor agent's bundle are valued at $-1$, removal of any single item gets rid of the inequity. Therefore, it is sufficient to establish \Cref{thm:EQ1n} for EQ1 allocations.


Note that since the instance is normalized, the reduced instance restricted to items in only $O^\pm$ is also normalized. That is, every agent assigns a value of $1$ to exactly $k_1$ items from $O^\pm$ and a value of $-1$ to exactly $k_2$ items from $O^\pm$, where $k_1$ and $k_2$ are constants. We show that there is a partial EQ1 allocation that allocates all the items in $O^\pm$ and then extend it to a complete allocation using \Cref{lem:completion}.

\begin{algorithm}[!ht]
\caption{$n$ agents, $\{1,-1\}$ Normalized Valuations}
\label{alg:EQ1_n}
\textbf{Input:} An instance with $n$ agents, $m$ items and $\{-1, 1\}$ Normalized Valuations \\
\textbf{Output:} An EQ1 allocation $A$
\begin{algorithmic}[1]
\STATE $A \gets$ An empty allocation 
\STATE $P \gets$ Set of all poor agents under $A$
\STATE $R \gets$ Set of all rich agents under $A$
\WHILE{$O^\pm \neq \emptyset$}
\WHILE{Any $p$ values an unallocated item $o$ in $O^\pm$ at $1$} \label{state:one}
\STATE $A_p = A_p \cup \{o\}$
\ENDWHILE
\WHILE{Any $r$ values an unallocated item $o$ in $O^\pm$ at $-1$}
\STATE $A_r = A_r \cup \{o\}$
\ENDWHILE
\STATE \label{state:eleven} Execute any of the feasible transfers described in the points (\ref{item:(a)}-\ref{item:(d)}) in \Cref{sec:symmetricvaluations}
\STATE Repeat from Step \ref{state:one}
\IF{none of the transfers in Step \ref{state:eleven} is feasible}
\IF{$|P| > |R|$}
\FOR{$j \in [t]$}
\STATE $A_{r_j} = A_{r_j} \setminus \{o : v_{r_j}(o) = 1\}$
\STATE $A_{p_j} = A_{p_j} \cup \{o\}$
\ENDFOR
\STATE Repeat from Step \ref{state:one}
\ELSE
\RETURN$A$
\ENDIF
\ENDIF
\ENDWHILE
\STATE Use \textbf{Completion Algorithm} of Lemma \ref{lem:completion}
\RETURN$A$ 
\end{algorithmic}
\end{algorithm}

To that end, 
we first describe \Cref{alg:EQ1_n}. We iteratively pick a poor agent and allocate it an item it values at $1$ (a good) from $O^\pm$. If there are no such items in $O^\pm$, then we iteratively pick one of the rich agents and allocate it an item it values at $-1$ (a chore) from $O^\pm$. At this point, it is easy to see that the partial allocation is EQ1. We denote the set of agents with minimum utility as poor agents $(P)$ and those with maximum utility as rich agents ($R$). Now suppose that the remaining items in $O^\pm$ are all chores for the poor agents and goods for the rich agents, and if they are allocated to any of them, they increase the amount of inequity and hence may violate EQ1. Therefore, in order to move towards an EQ1 allocation that allocates all the items in $O^\pm$, we now aim to convert a poor (rich) agent to rich (poor) by re-allocating one of the allocated items, so that the converted agent can now be a potential owner of one of the remaining unallocated items while maintaining EQ1. We do so by executing one of the following transfers at a time, each of which allows us to allocate at least one unallocated item while respecting EQ1.

\begin{enumerate} 
    \item \label{item:(a)} Rich-Rich Transfer (R-R). If there exist agents $r, r' \in R ~\text{and an item}~ o \in A_r ~\text{such that}~ v_r(o)=1 ~\text{and}~v_{r'}(o) = -1$, transfer $o$ from $A_r$ to $A_{r'}$. 
    \vspace{0.2cm}
    \item Rich-Poor Transfer (R-P). If there exist agents $ r\in R, p \in P~\text{and an item}~ o \in A_r~\text{such that}~  v_r(o)=1=v_p(o)$, transfer $o$ from $A_r$ to $A_{p}$. 
    \vspace{0.2cm}
    \item Poor-Rich Transfer (P-R). If there exist agents $r\in R, p \in P~\text{and an item}~ o \in A_p~\text{such that}~  v_p(o)=-1=v_r(o)$, transfer $o$ from $A_p$ to $A_r$. 
    \vspace{0.2cm}
    \item \label{item:(d)} Poor-Poor Transfer (P-P). If there exist agents $p, p'\in P~\text{and an item}~ o \in A_p~\text{such that}~ v_p(o)=-1 ~\text{and}~ v_{p'}(o)=1$, transfer $o$ from $A_p$ to $A_{p'}$. 
    \end{enumerate}

We execute one of the transfers at a time, thereby converting at least one poor (rich) agent to rich (poor). Now, the remaining item $o$, which was earlier a chore for all the poor agents and a good for all the rich agents, has a potential owner, and the algorithm makes progress by allocating $o$ to that agent.
We are now ready to prove the theorem.

\begin{table}[H]
\footnotesize
\centering
\begin{tabular}{c|cc|ccc}
    & \multicolumn{2}{c|}{$A_r$} & \multicolumn{3}{c}{$A_p$} \\ \hline
$r$ & \tikzmark{z}{$1$}      & \tikzmark{a}{$1$}           &         &         & \tikzmark{w}{\color{gray!50}{$-1$}}     \\ \hline
$p$ &          & \tikzmark{b}{{\color{gray!50}{$1$}}}               & $1$     & $1$     & \tikzmark{x}{$-1$}    
\end{tabular}
\link{a}{b}
\link{x}{w}
\caption{Rich-poor and poor-rich transfers. In one step, only one of these transfers is executed, not both.}
\label{tab:richpoortransfer}
\end{table}

\begin{proof}[Proof Sketch for \Cref{thm:EQ1n}] We show that $\Cref{alg:EQ1_n}$
returns an EQ1 allocation.
Suppose there are no feasible transfers. If $|P|>|R|$, then we take a good ($1$-valued item) from $A_r$ and allocate it to $p$ who necessarily values it at $-1$ (otherwise, we would have implemented R-P transfer). We do this for $|R|$ disjoint pairs of a rich and a poor agent, thereby decreasing the utility of all these pairs by $-1$. But since $|P|>|R|$, we have a poor agent $p'$ whose value remains intact, which in turn, makes him one of the rich agents after the above transfers. Now, the remaining unallocated item, which was a chore for $p'$ can be allocated to it without violating EQ1. When $|P| \leq |R|$, then we argue that no feasible transfers imply a contradiction by using the assumption of normalization. We defer this case to the appendix.
\end{proof}

If the instances are allowed to have $0$-valued items, then the approach of assigning $O^\pm$ first and then completing the allocation fails even for type-normalized valuations, as illustrated by the following example. Note that in this case, $O\pm$ may not be normalized even if the entire instance is normalized.

\begin{example}
\label{ex:nopartialEQ1}
There is no (partial) allocation $A$ of $O^\pm = \{o_1, o_2, o_3\}$ that satisfies EQ1.

  \begin{table}[H]
    \footnotesize
    \centering
    \begin{tabular}{c|ccc|ccc|ccc}
         & $o_1$ & $o_2$ & $o_3$ & $o_4$ & $o_5$ &$o_6$ & $o_7$ & $o_8$ & $o_9$ \\ \hline  
        Alice & $1$  & $1$  & $1$  & $0$ & $0$ & $0$ & $-1$ & $-1$ & $-1$ \\
        Bob & $-1$  & $-1$ & $-1$ & $1$ & $1$ & $1$ & $0$ & $0$ & $0$\\
    \end{tabular}
\end{table}
\end{example}

Therefore, when we have $\{-1, 0, 1\}$ valuations, we need a slightly different approach of allocating items from the entire bundle rather than exhausting a subset of items first. Using this approach, we present an efficient algorithm for finding EQ1 allocations. It is relevant to note here that in the presence of $0$ valued items, EQ1 allocations do not coincide with EQX allocations, since removal of a good valued at $0$ from the richer bundle is not helpful towards getting rid of the inequity.




\begin{restatable}{theorem}{thmEQonewow}
\label{thm:norm_EQ1}
For any mixed fair division instance with symmetric tri-valued normalized valuations ($\{-1,0,1\}$), an EQ1 allocation always exists and can be computed in polynomial time. 
\label{thm:EQnzero}
\end{restatable}

The proof idea is similar to that of $\Cref{thm:EQ1n}$. We continue with the greedy algorithm until every poor agent values the remaining items at $-1$ and every rich agent values them at $1$. We then execute the transfers described below. The claim is that either we arrive at an EQ1 allocation or we arrive at a contradiction by leveraging the assumption of normalization. We defer the complete proof to the appendix, but present the set of feasible transfers below.
After one of these transfers is implemented, at least one of the involved agents becomes a potential owner of an unallocated item without violating EQ1.

\begin{enumerate}
  \item \label{item:RRtransfer} Rich-Rich Transfer (R-R). If there exist agents $ r, r'\in R ~\text{and an item}~ o \in A_r~\text{such that}~(v_r(o), v_{r'}(o)) \in \{(1, 0), (0, -1), (1, -1)\}$, transfer $o$ from $A_r$ to $A_{r'}$. 

      \vspace{0.2cm}

  \item Rich-Poor Transfer (R-P). If there exist agents $ r\in R, p \in P~\text{and an item}~ o \in A_r~\text{such that}~ (v_r(o), v_{p}(o)) \in \{(1, 1), (1, 0), (0, 1)\}$, transfer $o$ from $A_r$ to $A_p$. 

      \vspace{0.2cm}

    \item Poor-Rich Transfer (P-R). If there exist agents $r\in R, p \in P~\text{and an item}~ o \in A_p~\text{such that}~ (v_p(o), v_r(o)) \in \{(-1, 0), (-1, -1), (0, -1)\}$, transfer $o$ from $A_p$ to $A_r$.

        \vspace{0.2cm}

    \item Poor-Poor Transfer (P-P). If there exist agents $ p, p'\in P~\text{and an item}~ o \in A_p~\text{such that}~ (v_p(o), v_{p'}(o)) \in \{(-1, 1),(0, 1), (-1, 0)\}$, transfer $o$ from $A_p$ to $A_{p'}$. 
\end{enumerate}

Note that the above result does not show the existence of EQX allocation, but does show the existence of a weaker notion of EQX$^+_-$ allocation \citep{10.24963/ijcai.2024/338}, that entails the removal of items with non-zero marginals, that is, either a strictly positive valued item or a strictly negative valued item. 






\subsection{Challenges Beyond Symmetric Tri-valued Valuations} We note here that in the presence of both subjective and objective items with general valuations, even with type-normalized instances, the intuitive idea of allocating a poor agent its most favorite item, if it exists, else allocating a rich agent its most disliked chore may not always work. One of the key challenges in such cases is that any poor agent may not derive positive utility from the remaining items and any rich agent may only derive positive utility from the remaining items. And unlike symmetric valuations, finding a feasible transfer may not be computationally efficient. 



In addition, a leximin++ allocation (defined below), which is known to satisfy not only EQ1 but stronger EQX property for certain restricted instances, fails to achieve EQ1 for mixtures, even with normalized valuations.
A leximin allocation is one that maximizes the minimum utility of an agent, subject to that maximizes the second-minimum utility, and so forth. Although there can be many leximin optimal allocations, they all induce a unique utility vector. Leximin++ allocation is an allocation that maximizes the minimum utility, and then maximizes the size of the bundle of the agent with the minimum utility, and so on. It is known that a leximin++ allocation is also EQX, for objective valuations with identical chores \citep{barman2024nearly}. When we go beyond objective valuations, Example \ref{ex:lemiminnotEQ1} shows that leximin++ may not be EQ1, even for a normalized instance with two agents and two items. Additionally,   \cite{10.1007/978-3-031-48974-7_12} showed that for $\{-1, 0, w\}$ (where $w$ is a positive integer) valuations, leximin allocations are EF1.

\begin{example} 
\label{ex:lemiminnotEQ1} Leximin++ is not EQ1 even for a normalized instance with two agents $(a_1, a_2)$ and two items $(o_1, o_2)$. $a_1$ values the items at $(10 , -15)$ while $a_2's$ valuation is $(-2, -3)$. There are at most $4$ possible allocations in this instance with utility vectors $(0, -5), (-5, 0), (10, -3), (-15, -2)$ for $(a_1, a_2)$ respectively. The only leximin++ allocation is the one with utility vector $(10, -3)$, but it is not EQ1.
\end{example}

Even if we restrict the valuations to $(1, -1)$, \Cref{ex:leximinPlusEQ1binary} suggests that leximin++ may not be EQ1. 

\begin{example} 
\label{ex:leximinPlusEQ1binary} Leximin++ is not EQ1 even with $\{1, -1\}$-normalized valuations. The instance in \Cref{tab:lexnoteq111} illustrates this. The highlighted allocation is leximin++ with the utility vector as $\{3, 1, 1, 1, 1\}$. This is not EQ1 since even if the last $5$ agents choose to hypothetically ignore a good from $a_1's$ bundle, they still fall short of the equitable utility.

\begin{table*}[ht]
\centering
\begin{tabular}{l|llllllllllll}
      & $o_1$ & $o_2$ & $o_3$ & $o_4$ & $o_5$ & $o_6$ & $o_7$ & $o_8$ & $o_9$ & $o_{10}$ & $o_{12}$ & $o_{13}$ \\ \hline
$a_1$ & $-1$    &\circled{$1$}     & \circled{$1$}    & \circled{$1$}    & \circled{$1$}    & \circled{$1$}     & \circled{-1}    & \circled{-1}    &$-1$   &$-1$      &$-1$      &$-1$      \\
$a_2$ & \circled{$1$}    &$-1$   &$-1$   &$-1$   &$-1$   &$-1$   &$-1$   &$-1$   & $1$     & $1$        & $1$        & $1$        \\
$a_3$ & $1$     &$-1$   & $-1$   &$-1$   &$-1$   &$-1$   &$-1$   &$-1$   & \circled{$1$}    &$1$       &$1$       &$1$       \\
$a_4$ &$1$    &$-1$   &$-1$   &$-1$   &$-1$   &$-1$   &$-1$   &$-1$   &$1$    & \circled{$1$}       &$1$       &$1$       \\
$a_5$ &$1$    &$-1$   &$-1$   &$-1$   &$-1$   &$-1$   &$-1$   &$-1$   &$1$    &$1$       & \circled{$1$}       &$1$       \\
$a_6$ &$1$    &$-1$   &$-1$   &$-1$   &$-1$   &$-1$   &$-1$   &$-1$   &$1$    &$1$       &$1$       & \circled{$1$}      
\end{tabular}
\vspace{0.4cm}
\caption{Leximin++ is not EQ1 even with \{1, -1\} normalized valuation. The highlighted allocation is leximin++ with the utility vector as $\{3, 1, 1, 1, 1\}$. This is not EQ1 since even if the last $5$ agents choose to hypothetically ignore a good from $a_1's$ bundle, they still fall short of the equitable utility.}
\vspace{0.4cm}
\label{tab:lexnoteq111}
\end{table*}

\end{example}


\subsection{Type-Normalized Valuations}

We now consider general valuations and present some tractable cases for type-normalized valuations.
\begin{theorem} 
\label{thm:type_n=2}
For any mixed fair division instance with type-normalized valuations and two agents, an EQ1 allocation always exists and can be computed in polynomial time.
\end{theorem}

\begin{proof} The algorithm starts by allocating items from $O^\pm$. For every $o \in O^{\pm}$, $o$ is allocated to the agent who values it positively. (Since $o \in O^\pm$ and $n=2$, exactly one of the agents values it positively). Once $O^{\pm}$ is allocated entirely, let $v_1$ and $v_2$ be the utilities derived by the two agents respectively, and say $v_1>v_2$. Then by type-normalization, there must be enough items in $O^+$ such that agent $2$ derives at least $v_1-v_2$ utility. The algorithm then allocates all such items to agent $2$, starting from its favorite good from the remaining items, until its utility becomes at least as much as agent $1$. Indeed, this partial allocation is EQ1. The remaining instance is the one that contains objective goods and chores. Using \Cref{lem:completion}, we get a complete EQ1 allocation, as desired.
\end{proof}

Although an EQ allocation may not exist even for two agents and one item, the assumption of type-normalization overcomes this non-existence for subjective items, as shown below.

\begin{proposition}
\label{thm:POisEQ}
For any mixed fair division instance with type-normalized valuations, two agents, and only subjective items, every PO allocation is EQ. Hence, an EQ allocation always exists and can be computed in polynomial time.
\end{proposition}
\begin{proof} Since $O^+ = O^- = \emptyset$ and $n=2$, an item that is positively valued by one agent is negatively valued by the other. Consider a PO allocation $A = \{A_1, A_2\}$. Under $A$, agent $1$ gets all the items that it values positively, and agent $2$ gets the remaining items (which are all positively valued by her, by the structure of the instance). Because of type-normalization $\left(\sum_{o \in G} v_i(o) = g\right)$, we have $v_1(A_1) = g = v_2(A_2)$, hence $A$ is EQ. 
\end{proof}

\section{EQ1 and Pareto Optimality}
\label{sec:EQ1$+$PO}

In this section, we discuss equitability in conjunction with PO. Recall that for goods, an EQ1$+$PO allocation may not exist in general but exists for strictly positive valuations; and deciding its existence in general is NP-Hard \citep{10.5555/3367032.3367073}. For chores, such an allocation always exists and admits a pseudo-polynomial time algorithm, but its complexity is an open question \citep{10.5555/3398761.3398810}. We now present our results for mixtures. We first show that an EQ1$+$PO allocation may not exist even for $\{1,-1\}$ type-normalized valuations, as illustrated by the following example.

\begin{example} 
\label{ex:pos}
 Consider the following instance. Any PO allocation must allocate $\{o_1, o_2, o_3\}$ to Alice and $\{o_4, o_5, o_6\}$ to Bob or Clara. At least one of them gets a utility of at most $1$ and hence violates EQ1 with respect to Alice. 
\begin{table}[!ht]
\footnotesize
    \centering
    \begin{tabular}{c|cccccc}
         & $o_1$ & $o_2$ & $o_3$ & $o_4$ & $o_5$ &$o_6$  \\ \hline 
        Alice & \circled{$1$}  & \circled{$1$}  & \circled{$1$}  & $-1$ & $-1$ & $-1$\\
       Bob & $-1$  & $-1$ & $-1$ & \circled{$1$} & \circled{$1$} & $1$  \\
        Clara & $-1$  & $-1$ & $-1$ & $1$ & $1$ & \circled{$1$}  \\
    \end{tabular}
\end{table}

\end{example}

Not only do we encounter the non-existence of EQ1$+$PO allocations, but deciding the existence of such allocations is hard even when the valuations are type-normalized, as shown by the following result. 

\begin{restatable}{theorem}{eqOnepohard}
\label{thm:eq1pohard}
        For any mixed fair division instance, deciding the existence of an EQ1$+$PO allocation is strongly NP-Hard.
\end{restatable}

\begin{proof} We present a reduction from $3$-Partition, known to be strongly NP-hard, where the problem is to decide if a multiset of integers can be partitioned into triplets such that all of them add up to a constant. Formally, the input is a set $S = \{b_1, b_2, \ldots b_{3r}\}$; $r\in \mathbb{N}$; and the output is a partition of $S$ into $r$ subsets $\{S_1, S_2, \ldots S_r\}$
such that $\sum_{b_i \in S_i} b_i = T = \frac{1}{r} \sum_{b_i \in S} b_i$. 
Given any instance of $3$-Partition, we construct an instance of allocation problem as follows. We create $r$ set-agents namely $\{a_1, a_2, \ldots a_r\}$ and one dummy agent, $a_{r+1}$. We create $3r$ set-items namely $\{o_1, o_2, \ldots o_{3r}\}$ and two dummy items $\{o_{3r+1}, o_{3r+2}\}$. All the set agents value the set items identically at $b_i$ and the two dummy items at $-T$ and $-(r-3)T$ respectively. The dummy agent values the set-items $\{o_1, \ldots o_{3r-1}\}$ at $0$, $o_{3r}$ at $-(r-2)T$ and the two dummy items at $T$ and $(r-1)T$ respectively. Note that this is a type-normalized instance such that every agent values all the items together at $2T$.
This completes the construction (see \Cref{tab:reductiontwo}). 

\begin{table}[H]
    \centering
    \begin{tabular}{c|cc|ccc}
         & $o_1$ $~\ldots~$ $o_{3r-1}$ & $o_{3r}$  & $o_{3r+1}$ & $o_{3r+2}$ \\ \hline 
       $a_1$  & $b_1$ $~\ldots~$  $b_{3r-1}$ & $b_{3r}$ & $-T$ & $-(r-3)T$  \\
     $a_2$ & $b_1$  $~\ldots~$  $b_{3r-1}$ & $b_{3r}$ & $-T$ & $-(r-3)T$  \\
       $a_r$  & $b_1$  $~\ldots~$  $b_{3r-1}$ & $b_{3r}$ &$-T$ & $-(r-3)T$ \\
       $a_{r+1}$ &$0$ $~\ldots~$ $0$ & $-r+2T$ & $T$ & $(r-1)T $ 
    \end{tabular}
    \vspace{0.3cm}
    \caption{Reduced instance as in the proof of \Cref{thm:eq1pohard}}\vspace{0.3cm}\label{tab:reductiontwo}
\end{table}

We now argue the equivalence.
\paragraph{Forward Direction.} Suppose $3$-Partition is a yes-instance and $\{S_1, S_2, \ldots S_r\}$ is the desired solution. Then, the corresponding EQ1$+$PO allocation, $A$, can be constructed as follows. For $a_i : i \in [r]$, we set $A(a_i) = \{o_j: o_j \in S_i\}$. Finally, $A(a_{r+1}) = \{o_{3r+1}, o_{3r+2}\}$. It is easy to see that $A$ is an EQ1 allocation since $v_i(A_i)) = T$ for all $a_i : i \in [r]$ and $v_i(A_i) \geq v_{r+1}(A(a_{r+1}) \setminus \{o_{3r+2}\}) = T$. Also, $A$ is PO since it is a welfare-maximizing allocation where each item is assigned to an agent who values it the most.


\paragraph{Reverse Direction.} Suppose that $A$ is an EQ1$+$PO allocation. Then, both the dummy items $\{o_{3r+1}, o_{3r+2}\}$ must be allocated to the dummy agent $a_{r+1}$, who is the only agent who values both the items positively. (Since any set-agent derives a negative utility from these items, allocating them to it violates PO). Since $A$ is also EQ1, this forces every set-agent to derive the utility of at least $T$ under $A$ so that when it removes the item $\{o_{3r+2}\}$ from $A(a_{r+1})$, EQ1 is preserved.
\end{proof}

 Towards tractability, we now show that an EQ1$+$PO allocation can be computed efficiently for $\{-1,0,1\}$ valuations (not necessarily normalized), if one exists. It is relevant to note here that the presence of zeros has been a source of computational challenge for EQ1$+$PO allocations even in the only goods setting \citep{10.5555/3367032.3367073}.
 We briefly describe the algorithm here. We first reduce the partial instance $\mathcal{I} := O^\pm \cup O^+$ to a goods-only instance $\mathcal{I^G}$ by setting $v_i(o) = 0$ for all $\{o \in O^\pm: v_i(o)=-1\}$. We then compute the Nash optimal allocation\footnote{Note that the Nash optimal solution maximizes the geometric mean of agents' utilities. It can be computed for binary valuations in polynomial time \citep{10.5555/3237383.3237392}.} $A'$ for $\mathcal{I^G}$. We extend the partial allocation $A$ to a complete allocation by first allocating the chores to an agent who values them at $0$ (if such an agent exists). Otherwise, we use \Cref{lem:completion} to allocate the remaining chores which are all valued at $-1$ by all the agents. If the completion of $A'$ is EQ1$+$PO, \Cref{alg:EQ1_PO_n} outputs the complete allocation, otherwise the instance does not admit any EQ1$+$PO allocation.
 

\begin{restatable}{theorem}{EQonePOwow}
\label{lem:typen101}
    For any mixed fair division instance with symmetric tri-valued (even non-normalized) valuations ($\{-1, 0, 1\}$), an EQ1$+$PO allocation can be found efficiently, if one exists.
\end{restatable}

\begin{proof}
We show that \Cref{alg:EQ1_PO_n} returns an EQ1$+$PO allocation whenever it exists. Towards correctness, we first establish the following result.

\begin{lemma}
\label{lem:NashEQ1PO}
     There is an EQ1$+$PO allocation for the partial instance $\mathcal{I}:= O^\pm \cup O^+$ if and only if the Nash optimal allocation $A'$ for the instance $\mathcal{I^G}$ is EQ1.
\end{lemma}

\begin{proof} 
Suppose the Nash optimal allocation $A'$ for instance $\mathcal{I^G}$ is EQ1. We will show that $A'$ is EQ1$+$PO for $\mathcal{I}$ as well. Note that $A'$ is EQ1$+$PO for $\mathcal{I^G}$ (since Nash satisfies PO). Since $A'$ is PO for $\mathcal{I^G}$, it never allocates an item to an agent who values it at $0$. Therefore, in the instance $\mathcal{I}$, $A'$ never allocates any item from $O^\pm \cup O^+$ to an agent who values it at $-1$ or $0$. Hence, all items are allocated to agents who value them the most, and hence $A'$ is PO (and EQ1) for $\mathcal{I}$. 

On the other hand, suppose there is an EQ1$+$PO allocation $A$ for instance $\mathcal{I}$. If the Nash optimal allocation for $\mathcal{I^G}$ is EQ1, we are done. Otherwise, suppose that the Nash optimal allocation for $\mathcal{I^G}$ is not EQ1. We will now argue by contradiction that this is not possible. Since $A$ is PO for the instance $\mathcal{I}$, any item $o \in O^\pm$ must have been allocated to an agent $a$ who values it at $1$. Indeed, if not, then there is a Pareto improvement by allocating $o$ to $a$, which makes $a$ strictly better off without making any other agent worse off. Likewise, all items $o \in O^+$ are allocated to respective agents who value them at $1$. So, we have $\sum_{i \in N} v_i(A_i) = m'$ where $m' = |O^\pm \cup O^+|$. Therefore, $A$ is not only PO but also achieves the maximum utilitarian welfare in $\mathcal{I}$. Consider the same allocation $A$ in the instance $\mathcal{I^G}$. Since the only modification from $\mathcal{I}$ to $\mathcal{I^G}$ is that for the agents who valued items in 
$O^\pm$ at $-1$ in $\mathcal{I}$, now value them at $0$ in $\mathcal{I^G}$ and everything else remains the same. So, we have $\sum_{i \in N} v_i(A_i) = m'$ in $\mathcal{I^G}$ as well. Therefore, $A$ maximizes the utilitarian welfare in $\mathcal{I^G}$ and hence, is PO. Also, $A$ is EQ1 by assumption. Therefore, for the instance $\mathcal{I^G}$, $A$ is an EQ1$+$PO allocation. But, the following result by \cite{freeman2019arxiv} shows that if an instance with binary valuations admits some EQ1 and PO allocation, then every Nash optimal allocation must satisfy EQ1.

\begin{lemma}[Lemma 22, \citep{freeman2019arxiv}] Let $\mathcal{I}$ be a fair division instance with binary valuations. If  $\mathcal{I}$ admits some EQ1 and PO
allocation, then every Nash optimal allocation must satisfy EQ1.
\end{lemma}

Therefore, the Nash optimal allocation in $\mathcal{I^G}$ satisfies EQ1, which is a contradiction to our assumption.
\end{proof}

If $A$ is EQ1$+$PO for $\mathcal{I}$, then we show that the allocation of $O^-$ as mentioned above maintains EQ1$+$PO. If $A$ is not EQ1$+$PO, then we show that if the EQ1 violators under $A$ can not be resolved by using the remaining chores, then there is no EQ1$+$PO allocation for the original instance. We defer the details to the appendix.
\end{proof}

\begin{algorithm}[!h]
\caption{Computing an EQ1$+$PO allocation, if it exists}
\label{alg:EQ1_PO_n}
\textbf{Input:} $n$ agents, $m$ items and $\{-1, 0, 1\}$ valuations \\
\textbf{Output:} An EQ1$+$PO allocation, if one exists \\
\begin{algorithmic}[1]
    \STATE $A \gets$ An empty allocation 
    \STATE For $o \in O^\pm$ such that $v_i(o) = -1$, set $v_i(o)=0$
    \STATE 
 Let $A'$ be the Nash optimal allocation on $O^{\pm} \cup O^+$.
    \STATE $A_i \gets A_i \cup A'_i$ 
     \WHILE{$\exists o \in O^-: v_i(o) = 0$ for some agent $i$}
    \STATE $A_i \gets A_i \cup \{o\}$
    \ENDWHILE \label{step:step10} 
    \WHILE{$\exists o \in O^-$}
    \STATE $i \gets$ a rich agent under $A_i$
    \STATE $A_i \gets A_i \cup \{o\}$
    \ENDWHILE
    \RETURN $A$ 
\end{algorithmic}
\end{algorithm}

We now show that EQ1$+$PO allocation always exists for any number of agents with identical valuations and can be computed efficiently. We also get existence and efficient computation for two agents and $\{-1,0,1\}$ type-normalized valuations.

\begin{proposition}
    For any mixed fair division instance with identical valuations, an EQ1$+$PO allocation always exists and admits a polynomial time algorithm. 
\end{proposition}

\begin{proof} Since identical valuations are a subset of objective valuations, we use the algorithm of \Cref{thm:obj} to obtain an EQ1 allocation $A$. We claim that $A$ is also PO. Notice that for identical valuations, any complete allocation is PO. Indeed, consider any complete allocation $A^\prime$ that allocates a good $g$ to an agent $i$. Under any reallocation of $g$, $i$ becomes strictly worse off. On the other hand, consider any chore $c$ allocated to some agent. Then, under any reallocation of $c$ to any agent $j$, the receiving agent $j$ becomes strictly worse off. Therefore, since $A$ is a complete allocation, it is also PO. This settles our claim.
\end{proof} 

\begin{restatable}{theorem}{eqonepowowtype}
\label{lem:type101}
    For any mixed fair division instance with two agents and type-normalized $\{-1, 0, 1\}$ valuations, an EQ1$+$PO always exists and can be computed efficiently.
\end{restatable}




When we increase the number of agents from $2$ to $3$, an EQ1$+$PO allocation may not exist even with binary type-normalized valuations. Consider $3$ agents and $6$ items such that two of them value the items at $\{1,1,1, 0,0,0\}$ and the last agents values at $\{0,0,0,1,1,1\}$. Here, every PO allocation violates EQ1. We note here that the non-existence of EQ1$+$PO allocation for the goods-only setting is due to the presence of items valued at $0$ and such allocation always exists for strictly positive valuations \citep{10.5555/3367032.3367073}. But for mixtures, Example \ref{ex:pos} suggests that even if there are no zero-valued items, an EQ1$+$PO allocation may not exist.





\section{EQX and Social Welfare}
\label{sec:EQ1+welfare}
In this section, we discuss the computational complexity of finding allocations that maximize the utilitarian welfare (UW) or egalitarian welfare (EW) within the set of EQX allocations for mixtures (whenever such allocations exist). 
%
Even in the case when all the underlying items are goods, finding a UW or EW allocation within the set of EQ1 allocations (UW/EQ1 or EW/EQ1), for a fixed number of agents, is weakly NP-Hard \citep{10.5555/3545946.3598685}. This rules out the possibility of a polynomial-time algorithm for mixed instances. Nonetheless, we present pseudo-polynomial time algorithms for finding such allocations in mixed instances with a constant number of agents. 
Our algorithmic technique extends those developed by \citep{AZIZ2023773} to the mixed setting, and relies on dynamic programming. 

\begin{restatable}{theorem}{UWEQOneDP}
\label{thm:psuedo}
    For any mixed fair division instance with a constant number of agents, computing a Utilitarian or Egalitarian maximal allocation within the set of EQX allocations (whenever they exist) admits a pseudo-polynomial time algorithm.
\end{restatable}

\begin{proof} We present a dynamic programming algorithm that keeps a set of states representing the set of possible allocations. At each state, it considers allocating the item $o_k$ to one of the $n$ agents. Finally, it chooses the state that optimizes for social welfare and respects EQX. We denote $\sum_{o \in O} \max_i v_i(o) = V_g$ and  $\sum_{o \in O} \min_i v_i(o) = -V_c$. The states are of the form $(k, \mathbf{v}, \mathbf{g}, \mathbf{c})$ where $k \in [m], v_i \in [-V_c, V_g], ~g_i, c_i \in [m] ~\forall i \in [n]$.  The items $g_i$ and $c_i$ refer to the least valuable good and the least disliked chore in the bundle of agent $i$. The state 
$(k,\mathbf{v}, \mathbf{g}, \mathbf{c}):=$ True if and only if there is an allocation of objects $\{o_1, \ldots, o_k\}$ such that the value of agent $i$ is at least $v_i$ and the bundle of $i$ contains $g_i$ and $c_i$ as the least valuable good and the least disliked chore respectively. The initial state $(0; 0,0 \ldots 0;\mathbf{0}, \mathbf{0})$ refers to the empty allocation and is vacuously true. Consider the case when $o_1$, which is a good for agent $i$, is allocated to $i$. Then, the state $(1, 0, \ldots v_i(o_1), \ldots 0; \emptyset, \ldots o_1, \ldots \emptyset; \emptyset \ldots \emptyset)$ is True and every other state is False. The state that corresponds to the allocation of $o_k$ to some agent, say $i$ is given as follows:  
If $g_i \neq o_k$ and $c_i \neq o_k ~\forall i \in [n]$, then
\begin{align}  
\label{eq:eq1}
(k,v_1, \ldots, v_i, \ldots v_n; g_1,\ldots,g_n; c_1, \ldots, c_n) = \nonumber \\
\lor_{i \in [n]} (k-1,v_1, \ldots, v_i - v_i(o_k), \ldots v_n; \nonumber \\ g_1,\ldots g_n; c_1,\ldots,c_n)
\end{align} 

where if $o_k$ is a good for agent $i$ then $v_i(g_i) \leq v_i(o_k)$ else, $v_i(c_i) \geq v_i(o_k)$. Else, if $g_i = o_k$, then, 
\begin{align} 
\label{eq:eq2}
(k,v_1, \ldots v_n; g_1,\ldots, g_i=o_k, \ldots g_n; c_1, \ldots, c_n) = \nonumber \\ \lor_{g \in [o_1, o_{k-1}]} (k-1,v_1, \ldots, v_i - v_i(o_k), \ldots v_n; \nonumber \\ g_1,\ldots ,g, \ldots g_n; c_1, \ldots, c_n) 
\end{align}


where $g \in \cup_{r \in [k-1]} o_r$ and $v_i(o_k) \leq v_i(g)$.

Else, if $c_i = o_k$, then, 
\begin{align}
(k,v_1, \ldots, v_n; g_1,\ldots, g_n; c_1, \ldots c_i=o_k, \ldots c_n) = \nonumber \\ \lor_{c \in [o_1, o_{k-1}]} (k-1,v_1, \ldots, v_i + v_i(o_k), \ldots v_n; \nonumber \\ g_1,\ldots ,g_i, \ldots g_n; c_1, \ldots c, \ldots c_n) 
 \end{align}

where $c \in \cup_{r \in [k-1]} o_r$ and $v_i(o_k) \geq v_i(c)$.
 The states $(m,\mathbf{v}, \mathbf{g}, \mathbf{c})$ correspond to the final allocation. An allocation corresponding to one of the final states is EQX iff $v_i \geq v_j - v_j(g_j)$ and $v_i - v_i(c_i) \geq v_j$ for every pair $i, j \in [n]$. We defer the correctness and runtime analysis to the appendix. 
\end{proof}

\section{Conclusion} 
We show that an EQ1 allocation may not exist for mixtures of goods and chores, and finding an EQ1 allocation is NP-Hard, unlike the goods-only and chores-only settings. 
We present efficient algorithms for $\{-1,0,1\}$ normalized valuations.  Previously, EQX allocations were known to exist only for $2$ agents and objective valuations, while the case of subjective valuations remained computationally intractable even with $2$ agents. We make progress by presenting an efficient algorithm that outputs an EQX allocation for $\{-1, 1\}$ normalized valuations. 

Further, we show the hardness of EQ1$+$PO allocations and complement it by efficient algorithms for restricted valuations. 
Finally, we present a pseudo-polynomial time algorithm for finding welfare-maximizing EQX allocations.

Deciding the existence of EQ1 allocations under normalized/type-normalized valuations for more than two agents and general valuations stands as an interesting open question. An immediate open direction is to resolve the EQ1 existence and computation for instances with valuations ternary valuations ($\{-\alpha, 0, \beta\}$ where $\alpha \neq \beta$). Such valuations have been considered in the context of EF1 \citep{Bhaskar2024Trilean} and EF1+PO \citep{10.1007/978-3-030-58285-2_1}.


\bibliographystyle{plainnat}
\bibliography{References}

\clearpage
\appendix

\section*{Technical Appendix}

\section{Additional Related Work}
\label{sec:additionalrelatedwork}
As mentioned previously, Aleksandrov \citep{Alek2020} termed equitability as jealousy-freeness (JF), and defined its up-to-one-style-relaxations JFX and JF1 wherein, in addition to requiring that the hypothetical removals are restricted to non-zero valued items, these definitions also allow hypothetical addition of an item to a bundle. 
Aleksandrov \citep{Alek2023} segregated the notions where addition/duplication is allowed and renamed these notions to 
DJFX and DJF1 (that is, JFX and JF1 up to removal or duplication). Their JFX definition is analogous to our EQX definition, albeit they restrict the removals to non-zero valued items. However, we note that their Theorem 3 does not hold true. In particular, Algorithm $1$ fails to achieve JFX. Reason being, Algorithm $1$ proceeds in the following greedy manner: until all goods are allocated, it
picks the least utility agent and gives them their most
preferred remaining good. Then, until all chores are allocated, it picks the greatest utility agent and gives
them their least preferred remaining chore. Our \Cref{ex:notEQX} shows that such an algorithm does not produce an EQX allocation (which coincides with JFX allocation in this case) even with identical valuations and two agents. Removal of any non-zero valued 
chore from the bundle of Alice (who is jealous of Bob) works to get rid of the inequitability, however, removal of any non-zero valued good from Bob's bundle still results in Bob being the richer agent.

\section{Omitted Details of \Cref{sec:EQ1}}


\propobj*

\begin{proof} The algorithm iteratively picks a poor agent---breaking ties arbitrarily---and allocates it the most valuable good among the remaining items in $O^+$ according to its preference. Once $O^+$ is exhausted, it picks a rich agent who then receives its most disliked item from $O^-$. 

The correctness of the above algorithm is as follows. In the beginning, when no one is allocated any item, EQ1 is satisfied vacuously. We argue that if the allocation is EQ1 before the assignment of an item, it remains EQ1 after that assignment as well. Let $i$ be a poor agent ($v_i(A_i) \leq v_j(A_j) ~\forall~ j \neq i$) at iteration $t$. Let $g$ be the good most valued by $i$ among the remaining goods in $O^+$, which is added to $A_i$ in the iteration $t+1$. Now either $i$ continues to be a poor agent, in which case, the allocation continues to be EQ1. Otherwise, consider an agent $j$ such that $v_i(A_i \cup \{g\}) > v_j(A_j)$. Then, since $g$ was the last added good in $A_i$, it is also the least favorite item of $i$ in its entire bundle $A_i$ (note that nothing has been allocated from $O^-$ till this iteration). This implies that $v_i(A_i \cup \{g\} \setminus g) = v_i(A_i) < v_j(A_j)$. Hence, any other agent violates EQ with respect to agent $i$ only up to the recent addition of $g$ into $A_i$. Since the empty allocation in the beginning is vacuously EQ1, this settles the claim that the allocation is EQ1, until $O^+$ is exhausted. Now, consider when everything from $O^+$ has been allocated. The instance now reduces to the one with only chores. Suppose agent $i$ is a rich agent at this point, that is, $v_i(A_i) \geq v_j(A_j) ~\forall~j \in [n]$. The algorithm picks $i$ and allocates it a chore $c$ that it dislikes the most. If $v_i(A_i \cup c) \geq v_j(A_j)~\forall~j \in [n]$, then the allocation continues to be EQ1. Else, $v_i(A_i \cup c) < v_j(A_j)$ for some agent $j$. Since $i$ was a rich agent previously and hence became a potential recipient of $c$, therefore, $i$ can choose to hypothetically remove $c$ from his bundle in order to value its own bundle more than $j$. That is, $v_i(A_i \cup c \setminus \{c\}) = v_i(A_i) \geq v_j(A_j)$. This implies that the allocation remains EQ1 after the allocation of $c$. This settles the claim.
\end{proof}

\thmEQoneww*

We show that \Cref{alg:EQ1_n} returns an EQ1 allocation when the instance has $\{-1, 1\}$ normalized valuations. First, note that since the instance is normalized, the reduced instance restricted to items in only $O^\pm$ is also normalized. That is, every agent assigns a value of $1$ to exactly $k_1$ items from $O^\pm$ and a value of $-1$ to exactly $k_2$ items from $O^\pm$, where $k_1$ and $k_2$ are constants. The idea is to first allocate all the items in $O^\pm$ in an almost equitable way and then to extend the partial EQ1 allocation to a complete allocation using \Cref{lem:completion}. To that end, we first show that there exists a partial EQ1 allocation that allocates all items in $O^\pm$.

\begin{lemma}
\label{claim:EQ1forOpm} There exists a partial EQ1 allocation restricted to the items in $O^\pm$ such that it exhausts $O^\pm$. Also, such an allocation can be computed in polynomial time.
\end{lemma}

We first describe the algorithm. We iteratively pick a poor agent and allocate it an item it values at $1$ (a good) from $O^\pm$. If there are no such items in $O^\pm$, then we iteratively pick one of the rich agents and allocate it an item it values at $-1$ (a chore) from $O^\pm$. At this point, it is easy to see that the partial allocation is EQ1. We denote the set of agents with minimum utility as poor agents $(P)$ and those with maximum utility as rich agents ($R$). Now suppose that the remaining items in $O^\pm$ are all chores for the poor agents and goods for the rich agents, and if they are allocated to any of them, they increase the amount of inequity and hence may violate EQ1. Therefore, in order to move towards an EQ1 allocation that allocates all the items in $O^\pm$, we now aim to convert a poor (rich) agent to rich (poor) by re-allocating one of the allocated items, so that the converted agent can now be a potential owner of one of the remaining unallocated items while maintaining EQ1. We do so by executing one of the following transfers at a time, each of which allows us to allocate at least one unallocated item while respecting EQ1.

\begin{enumerate} 
    \item Rich-Rich Transfer (R-R). If $\exists~ r, r' \in R ~\text{and}~ o \in A_r: v_r(o)=1 ~\text{and}~v_{r'}(o) = -1$, transfer $o$ from $A_r$ to $A_{r'}$. This makes both $r$ and $r'$ poor agents and potential owners of $o \in U$.

    \item Rich-Poor Transfer (R-P). If $\exists~ r\in R, p \in P~\text{and}~ o \in A_r: v_r(o)=1=v_p(w)$, transfer $o$ from $A_r$ to $A_{p}$. This makes $r$ a poor agent and $p$ a rich agent. Consequently, both of them become a potential owner of $o \in U$.

    \item Poor-Rich Transfer (P-R). If $\exists~ r\in R, p \in P~\text{and}~ o \in A_p: v_p(o)=-1=v_r(o)$, transfer $o$ from $A_p$ to $A_r$. This makes $r$ a poor agent and $p$ a rich agent. Consequently, both of them become a potential owner of $o \in U$.

    \item Poor-Poor Transfer (P-P). If $\exists~ p, p'\in P~\text{and}~ o \in A_p: v_p(o)=-1 ~\&~ v_{p'}(o)=1$, transfer $o$ from $A_p$ to $A_{p'}$. This makes both $p$ and $p'$ rich agents and potential owners of $o \in U$.
    \end{enumerate}

    We execute one of the transfers at a time, thereby converting at least one poor (rich) agent to rich (poor). Now, the remaining item $o$, which was earlier a chore for all the poor agents and a good for all the rich agents, has a potential owner, and the algorithm makes progress by allocating $o$ to that agent. This continues until we can find a poor agent who can be converted to a rich one by such a transfer and can be allocated one of the remaining items, or until $O^\pm$ is exhausted, in which case we are done. 
    
    Now suppose no such transfers are feasible and $U \neq \emptyset$. At this point, we have the following valuations.

    \begin{enumerate}
        \item Every $r$ values every good in $A_{r'}$ at $1$.
        \item Every $p$ values every good in $A_r$ at $-1$.
        \item Every $r$ values every chore in $A_p$ at $1$
        \item Every $p$ values every chore in $A_{p'}$ at $-1$.
    \end{enumerate}

With these valuations at hand, we are now ready to prove \Cref{claim:EQ1forOpm}.

\begin{proof}[Proof of \Cref{claim:EQ1forOpm}]
We start with the greedy allocation, executing the
transfers one by one, allocating at least one unassigned item after every transfer, and then argue that until an unassigned good remains, one of the transfers can be executed. If not, then we leverage normalization to arrive at a contradiction. To that end, we consider the following cases.
\begin{enumerate} 
    \item \label{item:Case2a} If $|P|>|R|$, then we take a good ($1$-valued item) from $A_r$ and allocate it to $p$ who necessarily values it at $-1$ (else, there is an R-P transfer). We do this for $|R|$ disjoint pairs of a rich and a poor agent, thereby decreasing the utility of all these pairs by $-1$. But since $|P|>|R|$, we have a poor agent $p'$ whose value remains intact, which in turn, makes him one of the rich agents after the above transfers. Now, the remaining unallocated item, which was a chore for $p'$ can be allocated to it without violating EQ1. 
    
    \item \label{item:EgEc} Suppose $|P| \leq |R|$. We will argue that this assumption leads to a contradiction to normalization, thereby settling the claim.
    
    Suppose under the partial EQ1 allocation $A$, we have $v_r({A_r}) = k+1$ and $v_p({A_p}) = k$ for some constant $k$. 

    Since every rich agent $r$ values every good in the bundle $A_{r'}$ of any other rich agent $r'$ at $1$, we have that there are at least $|R|(k+1)$ allocated items that a rich agent $r$ values at $1$. Additionally, $r$ assigns a value of $1$ to each of the remaining unallocated items in $O^\pm$, meaning there are at least $|R|(k+1)+1$ items in total that $r$ values at $1$. Since all the unallocated items in $O^\pm$ are considered chores by agent $p$, it follows—by normalization—that at least $|R|(k+1)+1$ of the \emph{allocated} items must be valued at $1$ by $p$. 

    Even if $p$ values every good in every other poor agent's bundle at $1$, then also only $|P|k$ out of $|R|(k+1)+1$ such items are accounted for. (Note that $p$ can not value a chore in $A_{p'}$ at $1$, else there exists a P-P transfer). Also, note that $p$ cannot value any of the $|R|(k+1)$ items allocated to the rich agents at $1$, else there is an R-P transfer. This implies that there must be at least $(|R|-|P|)k + |R| +1$ many \emph{allocated} items valued at $1$ by $p$, outside of the $|R|(k+1)$ and $|P|k$ items allocated as goods to the rich and poor agents respectively. At this point, we say that $r$ stays ahead of $p$ in terms of the number of items valued as a good. 
    
    Since $v_r(A_r) = k+1$, in particular, $r$ derives its value of $k+1$ from a set of $k+1$ items valued at $1$ by itself. We call these $k+1$ items as \emph{contributing} goods in the bundle $A_r$. There can be additional items in $A_r$, however, they do not contribute to the utility of $r$. Any such additional good $g$ (chore $c$) in $A_r$ is \emph{non-contributing}, that is, it can be paired up with a chore $c$ (good $g$) in $A_r$ so that the final utility contribution of $g \cup c$ to $r$ is $0$. We refer to these pairs as non-contributing good-chore pair $(g, c) \in A_r$. Likewise, for any poor agent $p$, there are $k$ contributing goods in $A_p$, and the rest of the goods (chores) in $A_p$ can be paired up with a chore (good) in $A_p$ so their cumulative utility contribution is $0$.

    A non-contributing good-chore pair $(g, c) \in A_r$ can be interpreted by any poor agent $p$ either as a chore-good $(c', g')$ pair or as a chore-chore \((c', c'')\) pair. (Note that the good \(g\) in \((g, c)\) cannot be a good for any poor agent, else there exists an R-P transfer.) In the case where it is interpreted as $(c', c'')$, the poor agent $p$ falls even further behind agent $r$ in terms of the number of goods. The agent $r$ stays ahead of the poor agents by this additional $1$-valued item $g$. In the former case, we map the item $g$ in $(g, c) \in A_r$ to the item $c$, which is valued as a good $g'$ by $p$. We associate the item $g$ in $(g, c) \in A_r$ with the item $c$, which is valued as a good $g'$ by the poor agent $p$. For every such $g'$ valued as a good by a poor agent, the rich agent contains $g$ in its bundle. Consequently, the non-contributing pairs in a rich bundle are insufficient to compensate for the deficit in the number of items valued as goods by a poor agent.

    For $p$ to compensate for this deficit, it must value at least $(|R|-|P|)k+|R| +1$ many \emph{allocated} items, outside of the non-contributing pairs in any rich bundle, at $1$. Let's call this set of goods for $p$ as $E_g^p$. Since we have counted for the contributing goods in any poor bundle (at most $|P|k$ of them), for any $g \in E_g^p$, we must have that it is allocated to a poor agent $p'$ as a non-contributing pair. (Again, the non-contributing pairs allocated to $r$ are already counted for). Let $c$ be the corresponding non-contributing chore in $A_{p'}$. If $v_p(c) = 1$, then we have a P-P transfer. So, for all $p \in P$, we have $v_p(c) = -1$. Now, if any rich agent values $c$ at $1$, then it again stays ahead of the poor agents by this additional $1$ valued item (which is outside its bundle), and the process of compensating for the deficit does not terminate. Otherwise, $v_r(c) = -1$, but this is a contradiction to the fact that $c \in O^\pm$.
 \end{enumerate}
 Therefore, starting with a greedy allocation, executing the transfers one by one, allocating at least one unassigned item after every transfer, the algorithm terminates at an EQ1 allocation for the items in $O^\pm$. This settles the claim.
 \end{proof}

\begin{table*}
\centering
\begin{tabular}{c|cccccccccccccccc|}
\multicolumn{1}{l|}{} & \multicolumn{16}{c|}{$O^\pm$}                                                                                                                                                                                                                                                                                                                                                                                                                                                                                                                                   \\ \hline
\multicolumn{1}{l|}{} & \multicolumn{10}{c}{\begin{tabular}[c]{@{}c@{}}Allocated \\ Items\end{tabular}}                                                                                                                                                                                                                                                                            & \multicolumn{1}{l}{}           & \multicolumn{1}{l}{}     & \multicolumn{1}{l}{}           & \multicolumn{1}{l|}{}     & \multicolumn{2}{c|}{\begin{tabular}[c]{@{}c@{}}Unallocated\\ Items\end{tabular}} \\ \hline
                      & $o_1$                          & $o_2$                          & $o_3$                          & $o_4$                          & $o_5$                          & $o_6$                          & $o_7$                          & $o_8$                          & $o_9$                          & \multicolumn{1}{c|}{$o_{10}$}                         & \multicolumn{4}{c|}{\begin{tabular}[c]{@{}c@{}}Non-Contributing \\ Items\end{tabular}}                                                                 & $o$                                     & $o$                                    \\ \hline
$r_1$                 & \circled{$1$} & \circled{$1$} & \circled{$1$} & {\color{gray!50}{$1$}}                            & {\color{gray!50}{$1$}}                            & {\color{gray!50}{$1$}}                            &                                &                                &                               & \multicolumn{1}{c|}{}                               &         $\boxed{1,-1}$                       & \multicolumn{1}{c|}{}    & $-1$,{\color{red}{$-1$}}                          & \multicolumn{1}{c|}{} & $1$                                     & $1$                                    \\
$r_2$                 & {\color{gray!50}{$1$}}                            & {\color{gray!50}{$1$}}                            & {\color{gray!50}{$1$}}                          & \circled{$1$} & \circled{$1$} & \circled{$1$} &                                &                                &                                & \multicolumn{1}{c|}{}                               &                                & \multicolumn{1}{c|}{}    & $-1$,{\color{red}{$-1$}}                            & \multicolumn{1}{c|}{} & $1$                                     & $1$                                    \\
$p_1$                 & {\color{gray!50}{$-1$}}                           & {\color{gray!50}{$-1$}}                           & {\color{gray!50}{$-1$}}                           & {\color{gray!50}{$-1$}}                           & {\color{gray!50}{$-1$}}                          & {\color{gray!50}{$-1$}}                          & \circled{$1$} & \circled{$1$} & $1$                            & \multicolumn{1}{c|}{$1$}                            & {\color{gray!50}{$-1,\star$}} & \multicolumn{1}{c|}{} & $1$, {\color{red}{$-1$}} & \multicolumn{1}{c|}{} & $-1$                                    & $-1$                                   \\
$p_2$                 & {\color{gray!50}{$-1$}}                           & {\color{gray!50}{$-1$}}                          & {\color{gray!50}{$-1$}}                          & {\color{gray!50}{$-1$}}                           & {\color{gray!50}{$-1$}}                           & {\color{gray!50}{$-1$}}                          & $1$                            & $1$                            & \circled{$1$} & \multicolumn{1}{c|}{\circled{$1$}} & {\color{gray!50}{$-1,\star$}}                           & \multicolumn{1}{c|}{} & \boxed{$1$,\circled{{\color{red}{$-1$}}}}                          & \multicolumn{1}{c|}{} & $-1$                                    & $-1$                                  
\end{tabular}
\vspace{0.4cm}
\caption{A schematic of Case (2b) in the proof of \Cref{claim:EQ1forOpm}. $r_1$ and $r_2$ are two rich agents with utility $3$ each and $p_1$ and $p_2$ are two poor agents with utility $2$ each. The values in gray are forced, otherwise, there is a feasible transfer. The values in red depict the contradiction that there is an item $c$ such that $c \notin O^\pm$}
\vspace{0.4cm}
\label{tab:EQ1cc}
\end{table*}

We are now ready to prove \Cref{thm:EQ1n}.

\begin{proof}[Proof of \Cref{thm:EQ1n}]
Lemma \ref{claim:EQ1forOpm} gives us a partial EQ1 allocation $A$. 
Now we can allocate the items from $O^+$ to a poor agent iteratively and once $O^+$ is exhausted, we allocate the item from $O^-$ to a rich agent iteratively (\Cref{lem:completion}).
This ensures that $A$ is EQ1 at every step of the algorithm. Consequently, the complete allocation $A$ is EQ1. 
\end{proof}

\thmEQonewow*

The idea is similar to \Cref{thm:EQ1n} except here, we do not aim to exhaust $O^\pm$ first, instead, allocate the items from the entire bundle. 
This is so because a partial allocation restricted to $O^\pm$ may not satisfy EQ1 (see \Cref{ex:nopartialEQ1}).

We give a brief algorithm description which is along the lines of \Cref{alg:EQ1_n}, except for the transfers, which are more involved in this setting.
We iteratively pick a poor agent $p$ and allocate it an item valued at $1$ or $0$. If there is no such item for $p$, then we iteratively pick one of the rich agents and allocate it an item valued at $-1$ or $0$. At this point, it is easy to see that the partial allocation $A$ is EQ1. As previously, we denote the set of agents with minimum utility under $A$ as poor agents $(P)$ and those with maximum utility under $A$ as rich agents ($R$). Either we arrive at a complete EQ1 allocation by continuing in the above manner, or there are unassigned items which are chores for agents in $P$ and goods for agents in $R$. Allocating any of the remaining items only increases the inequity and hence violates EQ1. Therefore, in order to move towards a complete EQ1 allocation, we now aim to convert a poor (rich) agent to rich (poor) by re-allocating one of the allocated items, so that the converted agent can now be a potential owner of one of the remaining unallocated items while maintaining EQ1. We do so by executing one of the following transfers at a time, after which at least one of them becomes a potential owner of $o \in U$.

\begin{enumerate}
  \item Rich-Rich Transfer (R-R). If $\exists~ r, r' \in R ~\text{and}~ o \in A_r: (v_r(o), v_{r'}(o)) \in \{(1, 0), (0, -1), (1, -1)\}$, transfer $o$ from $A_r$ to $A_{r'}$. 

  \item Rich-Poor Transfer (R-P). If $\exists~ r\in R, p \in P~\text{and}~ o \in A_r: (v_r(o), v_{p}(o)) \in \{(1, 1), (1, 0), (0, 1)\}$, transfer $o$ from $A_r$ to $A_p$. 

    \item Poor-Rich Transfer (P-R). If $\exists~ r\in R, p \in P~\text{and}~ o \in A_p: (v_p(o), v_r(o)) \in \{(-1, 0), (-1, -1), (0, -1)\}$, transfer $o$ from $A_p$ to $A_r$. 

    \item Poor-Poor Transfer (P-P). If $\exists~ p, p'\in P~\text{and}~ o \in A_p: (v_p(o), v_{p'}(o)) \in \{(-1, 1),(0, 1), (-1, 0)\}$, transfer $o$ from $A_p$ to $A_{p'}$. 
\end{enumerate}

Note that we execute one of the transfers at a time, thereby converting at least one poor (rich) agent to rich (poor). Now, the remaining item $o$, which was earlier a chore for all the poor agents and a good for all the rich agents, has a potential owner, and the algorithm makes progress by allocating $o$ to that agent. This continues until we can find a poor agent who can be converted to a rich one by such a transfer and can be allocated one of the remaining items. 
    
Now suppose no such transfers are feasible and $U \neq \emptyset$. At this point, we have the following valuations.

    \begin{enumerate}
        \item Every $r$ values every good in $A_{r'}$ at $1$.
        \item Every $p$ values every good in $A_r$ at $-1$.
        \item Every $r$ values every chore in $A_p$ at $1$
        \item Every $p$ values every chore in $A_{p'}$ at $-1$.
    \end{enumerate}

\begin{proof}[Proof of \Cref{thm:EQnzero}]
With the above set of transfers, it is easy to verify that the analogous argument as in \Cref{thm:EQ1n} holds true in this setting as well. Starting with a greedy allocation, executing the transfers one by one, allocating at least one unassigned item after every transfer, the algorithm terminates at an EQ1 allocation for all the items. This settles the claim.
\end{proof}



\section{Omitted Details of \Cref{sec:EQ1$+$PO}}


\begin{algorithm}[!ht]
\caption{EQ1$+$PO, n=2 and $\{-1, 0, 1\}$ Type-Normalized Valuations}
\label{alg:eq1po2}
\textbf{Input:} An instance with $2$ agents, $m$ items and $\{1, 0, -1\}$ type-normalized valuations. \\
\textbf{Output:} An EQ1$+$PO allocation $A$.
\begin{algorithmic}[1]
\STATE $A \leftarrow$ An empty allocation 
\WHILE{$O^\pm \neq \emptyset$}
\STATE $A_i = A_i \cup \{o\}$, where $o \in O^\pm$ such that $v_i(o) = 1$ 
\ENDWHILE
\IF{$\exists~ o \in O^+$ such that $v_i(o) = 0$ for some $i$}
\STATE $A(j) = A(j) \cup \{o\}$ such that $j \neq i$
\ENDIF
\WHILE{$O^+ \neq \emptyset$}
\STATE $i \gets$ poor agent
\STATE $A_i = A_i \cup \{o\}$ where  $o$ is the most valuable good from $O^+$ for $i$
\ENDWHILE
\IF{$\exists~ o \in O^-$ such that $v_i(o) = 0$ for some $i$}
\STATE $A_i = A_i \cup \{o\}$
\ENDIF
\WHILE{$O^- \neq \emptyset$}
\STATE $i \gets$ rich agent
\STATE $A_i = A_i \cup \{o\}$ where $o$ is the most disliked chore from $O^-$ for $i$
\ENDWHILE
\RETURN $A$
\end{algorithmic}
\end{algorithm}

\EQonePOwow*

\begin{proof}{(continued)} We show here that if $A$ is not EQ1, then there is no EQ1$+$PO allocation for the instance. 
If the EQ1 violators in $A^\prime$ (the Nash allocation on the reduced instance $\mathcal{I_G}$) cannot be resolved by using the remaining chores, then there is no EQ1$+$PO allocation. We argue this claim by contradiction. Suppose there is a complete EQ1$+$PO allocation, say $A^\star$. Let $t$ be the number of items in $O^-$, and hence valued at $-1$ by all the agents (they are `universal' chores). Under the Nash optimal partial allocation $A^\prime$, let $v_p = \min_i v_i(A^\prime_i)$ be the utility of the poorest agent and denote the set of all such agents as $P$, let the agents with utility $v_r = v_p + 1$ be called the rich agents, denoted as $R$, and the remaining agents with utility strictly greater than $v_r$ be called the EQ1 violators, denoted by $S$. It is easy to see that if $\sum_{s \in S} (v_s - v_r) \leq t$, then the $t$ many $-1$ valued chores
from $O^-$ could have been used to bring down the utility of all the violators to $v_r$, ensuring that the completion of  $A^\prime$ into $A$ is indeed EQ1. Therefore, we assume 

\begin{align}
\label{eq:eqNashnotEQ1}\sum_{s \in S} (v_s - v_r) >  t\end{align}

Now consider the EQ1$+$PO allocation $A^\star$. Let $\overline{A^\star}$ denote the restriction of $A^\star$ to $O^\pm \cup O^+$. Then, $\overline{A^\star}$ must be a PO allocation but not EQ1, otherwise the Nash optimal allocation must have been EQ1 (\Cref{lem:NashEQ1PO}). Let $v_p^\star,~ v_r^\star,~ v_s^\star$ denote the utilities of the poor agents $P$, rich agents $R$, and violators $S^\star$ in $\overline{A^\star}$. It must be the case that 
\begin{align} \label{eq:eqNashnotEQ1two} \sum_{s \in S^\star} (v_s^\star - v_r^\star) \leq t\end{align}
since the completion $A^\star$ is an EQ1 allocation. Also, notice that since both $\overline{A^\star}$ and $A^\prime$ are Pareto optimal, they both allocate the items in $O^\pm \cup O^+$ non-wastefully. Therefore, the sum of the utilities under both the partial allocations is exactly $m^\prime$, where $m^\prime:= |O^\pm \cup O^+|$. Therefore, 

\begin{align}\sum_{p \in P, ~r \in R, ~s \in S} (v_p + v_r + v_s) = \sum_{p \in P^*, ~r \in R^*, ~s \in S^*} (v_p^\star + v_r^\star + v_s^\star) = m' \end{align}
 
 Since all the $v_i$'s and $v_i^\star$'s are non-negative (because of non-wasteful allocation of items in $O^\pm \cup O^+$ under both $A'$ and $\overline{A^\star}$), by AM-GM inequality\footnote{For any $n$ non-negative numbers $\{v_1, \ldots v_n\}$ arithmetic Mean is greater than the geometric mean (that is, $\frac{1}{n}\sum_{i \in [n]}{v_i} \geq (\prod_{i \in [n]} v_i)^\frac{1}{n}$)}, it follows that the product of a set of numbers with a constant sum has the higher outcome if they are equal/closer to each other rather than being further away ($\sum_{s \in S^\star} (v_s^\star - v_r^\star) \leq t < \sum_{s \in S} (v_s - v_r)$ ensures that $v_r^\star$ and $v_s^\star$ are closer to each other than $v_s$ and $v_r$ (Equations \ref{eq:eqNashnotEQ1} and  \ref{eq:eqNashnotEQ1two})). Therefore,

 \begin{align} NW(A^\prime)=\left(\prod_{p \in P, r \in R, s \in S}~v_p v_r v_s\right)^\frac{1}{n} < \left( \prod_{p \in P^*, r \in R^*, s \in S^*}~v_p^\star  v_r^\star v_s^\star \right)^\frac{1}{n} =  NW(\overline{A^\star})  \end{align}
 Hence, $A'$ is not a Nash optimal allocation. This settles the claim.
\end{proof}

\eqonepowowtype*

\begin{proof} We will show that the allocation $A$ as returned by the \Cref{alg:eq1po2} is EQ1$+$PO. Notice that $A$ allocates all items from $O^\pm \cup O^+$ non-wastefully. Also, $o \in O^-$ is allocated to an agent who values it at $0$, otherwise, if both the agents value it at $-1$, then PO is satisfied irrespective of which agent ends up receiving that item.
 Therefore, $A$ satisfies PO. We now argue that it also satisfies EQ1. Suppose $i$ is a poor agent by the end of Step $6$. Say, the utility of $i$ at this point is $k_1 \cdot c$ and that of agent $j$ is $k_2 \cdot c$. Since $i$ is the poor agent, therefore, $k_1 < k_2$. Then, there must be at least $k_2-k_1$ items in $O^+$ valued at $\{1, 0\}$ by $\{i, j\}$ (by type-normalization). By construction, all these $k_2-k_1$ items are allocated to $i$ under $A$. This compensates for the inequity experienced by $i$ so far. Now the allocation is extended by iteratively allocating a poor agent its most valuable item, which maintains EQ1 till $O^+$ is exhausted. 
 For the items in $O^-$ which are valued at $-1$ by both agents, the rich agent gets that item, which thereby maintains EQ1. For the remaining chores in $O^-$, at least one of the agents values them at $0$ and ends up receiving the same. This does not violate EQ1 as it does not change the utility of the agents. This settles the claim. 
\end{proof}

\section{Omitted Details of \Cref{sec:EQ1+welfare}}

\UWEQOneDP*

\begin{proof}{(continued)}
We argue here that every entry in the DP table is indeed computed correctly. To that end, we do an induction on the number of items allocated. For the base case, when one item is allocated to say, agent $i$, then only $i$ derives the value of $v_i(o_1)$, and the rest of the agents get a value $0$. Depending on whether $o_1$ is a good or a chore for agent $i$, either $g_i = o_1$ or $c_i = o_1$, and everything else is $\emptyset$. This is correctly captured in the base case. By induction hypothesis, suppose all the table entries that allocate the first $k-1$ items are computed correctly. Consider the allocation of $k^{th}$ item as captured by the table entry $(k, \mathbf{v}, \mathbf{g}, \mathbf{c})$. 

First consider the case when $g_i \neq o_k$ and $c_i \neq o_k ~\forall i \in [n]$ in the entry $(k, \mathbf{v}, \mathbf{g}, \mathbf{c})$. Suppose RHS of \Cref{eq:eq1} is True. Suppose $o_k$ is a good for some agent $i$ such that $v_i(g_i) \leq v_i(o_k)$ (the case of chores can be argued similarly) and RHS of \Cref{eq:eq1} is True for the index $i$. That is, $(k-1, v_1, \ldots, v_i - v_i(o_k), \ldots v_n; g_1,\ldots g_i, \ldots g_n; c_1, \ldots c_i, \ldots c_n)$ is True. It means that there is allocation of $k-1$ items such that everyone gets a utility of $(v_1 \ldots v_i - v_i{(o_k)}, \ldots v_n)$ and $(g_1 \ldots g_n)$ and $(c_1, \ldots c_n)$ are the least valued goods and the least disliked chores in the respective bundles.  Then allocating $o_k$ to $i^{th}$ agent gives an allocation with utilities $(v_1, \ldots v_i, \ldots v_n)$ such that the set of the least valued goods and least disliked chores remain the same for all the agents (because $v_i(g_i) \leq v_i(o_k)$). Therefore, LHS of \Cref{eq:eq1} is True.

On the other hand, suppose the LHS of \Cref{eq:eq1} is True. Suppose $o_k$ is a good and belongs to agent $i's$ bundle. Then, since $g_i \neq o_k$, it means that $g_i$ is the least valued item in agent $i's$ bundle. Consider the allocation after removing $o_k$ from $i's$ bundle. Then, it corresponds to an allocation of $k-1$ items such that each agent gets a utility of $(v_1, \ldots v_i - v_i{(o_k)}, \ldots v_n)$ such that $(g_1, \ldots g_n)$ and $(c_1, \ldots c_n)$ are the least valued goods and the least valued chores in the respective bundles. Therefore, $(k, v_1, \ldots v_i - v_i{(o_k)}, \ldots v_n, g_1, \ldots g_n, c_1, \ldots c_n )$ is True and hence RHS of \Cref{eq:eq1} is True.

If $o_k$ is a good for agent $i$ such that $v_i(g_i) > v_i(o_k)$, RHS of \Cref{eq:eq1} is False by definition, and LHS of \Cref{eq:eq1} is False since $g_i$ was the least valuable item in $i's$ bundle but $g_i \neq o_k$.

Now consider the case when $g_i = o_k$ for some $i \in [n]$ in the table entry $(k, \mathbf{v}, \mathbf{g}, \mathbf{c})$. This means that $(k, \mathbf{v}, \mathbf{g}, \mathbf{c})$ corresponds to a state where $o_k$ is allocated to $i$ who considers it to be a good. Now suppose RHS of \Cref{eq:eq2} is True. Then, adding $o_k$ to $i's$ bundle increases its utility to $v_i$ from $v_i - v_i(o_k)$ and since $v_i(o_k) < v_i(g)$, $o_k$ is the new least valued item in $i's$ bundle, which is captured in the LHS and hence LHS is True. On the other hand, if LHS of \Cref{eq:eq2} is True, then consider the allocation post removing the item $o_k$ from $i's$ bundle. Then, RHS of \Cref{eq:eq2} is True for a least valued item $g$ in $i's$ bundle from $\{o_1, \ldots o_{k-1}\}$. The case when $c_i = o_k$ for some $i \in [n]$ is argued similarly. This settles the claim that the table entries are correctly computed at every step.

The states $(m,\mathbf{v}, \mathbf{g}, \mathbf{c})$ correspond to the final allocation. An allocation corresponding to one of the final states is EQX if and only if $v_i \geq v_j - v_j(g_j)$ and $v_i - v_i(c_i) \geq v_j$ for every pair of agents $i, j \in [n]$. Among the states that correspond to EQX allocations, the algorithm selects the one that maximizes UW, that is, $\sum_{i \in [n]}{v_i}$ or EW, that is, $\min_{i \in [n]}{v_i}$.

The total number of possible states is $m^{2n+1} \cdot V^n$, where $V= V_g + V_c$. Computing one state requires at most $n$ many look-ups of previously computed states. Therefore, the run time of the algorithm is $O(n \cdot m^{2n+1} \cdot V^n)$. Note that when valuations are $\{1,0, -1\}$, $V$ is a constant, hence, for a fixed number of agents, the algorithm runs in polynomial time for this case.
\end{proof}


\section{EF+EQ+PO Allocations}
\label{sec:EF+EQ}

In this section, we discuss the compatibility of equitability with envy-freeness, defined as follows.

\paragraph{Envy-Freeness.} An allocation $A$ is said to be \emph{envy-free} (EF) if for any pair of agents $i,j \in [n]$, we have $v_i(A_i) \geqslant v_i(A_j)$. It is said to be \emph{envy-free up to one item} (EF1) if for any pair of agents $i$ and $j$ such that $v_i(A_i) < v_j(A_j)$, either there is a good $g \in A_j$ such that $v_i(A_i) \geq v_i(A_j \setminus g)$ or there is a chore $c$ such that $v_i(A_i \setminus c) \geq v_i(A_j)$. Further, it is EFX  we have $v_i(A_i) \geq v_i(A_j \setminus \{g\})$ for all goods $g$ in $A_j$, and $v_i(A_i \setminus \{c\}) \geq v_i(A_j)$ for all chores $c$ in $A_i$. 


\begin{proposition}
\label{prop:EQ+POisEF}
    For $\{-1, 0, 1\}$-valuations, an EQ+PO allocation is also envy-free (EF).

\end{proposition}

\begin{proof} Suppose $A$ is an EQ+PO allocation for the given instance. EQ implies that we have $v_i(A_i) = v_i(A_j) = k \cdot c$ and PO ensures that if an agent receives an item $o$ with value $-1$, then everyone else values $o$ at $-1$. (Else, if there is an agent $i$ such that $v_i(o) = 0 ~or~c$, then allocating $o$ to $i$ is a Pareto improvement). Likewise, if an agent receives an item that it values at $0$, then everyone else values that item at either $0$ or $-1$, again for the same reason. Now suppose $A$ is not EF. Then, there is a pair of agents $i$ and $j$ such that $v_i(A_i) = k \cdot c$ but $v_i(A_j)> k \cdot c = v_j(A_j)$. This implies that there is an item $o$ in $j's$ bundle that is valued at $0$ (or $-1$) by $j$ but valued at $1$ (or $0$) by $i$. Allocating $o$ to $i$ is a Pareto improvement, which is a contradiction. Therefore, $A$ is EF.
\end{proof}

Since the allocation constructed in the proof of \Cref{thm:EQ1hard} is EF, therefore, we get the following result.

\begin{corollary}
\label{cor:EQ1+EFhard} Deciding whether an instance admits an allocation that is simultaneously EF+EQ1 or EF1+EQ1 is (weakly) NP-complete.
\end{corollary}

Since the allocation constructed in the proof of \Cref{thm:eq1pohard} is EF, therefore, we get the following result.

\begin{corollary}
\label{cor:EQ1+EF+POhard}
Deciding whether a type-normalized instance admits an allocation that is simultaneously EF+EQ1$+$PO or EF1+EQ1$+$PO is (strongly) NP-hard.
\end{corollary}

Notice that for $\{1,0,-1\}$ valuations, a Pareto optimal allocation is EF1 if and only if it is EFX, and is
EQ1 if and only if it is EQX. Therefore, the above two results hold for
all combinations of X + Y + PO, where X $\in$ \{EFX, EF1\} and Y $\in$ \{EQX, EQ1\}. The allocation constructed in \Cref{lem:typen101} can be easily verified to be EF1, thereby confirming the following result.

\begin{corollary}
\label{cor:EQ1+EF1+POeasy}
For $\{1, 0,-1\}$ valuations, an EF1+EQ1$+$PO allocation can be computed in polynomial time, whenever such an allocation exists.
\end{corollary}

\end{document}